\documentclass[a4paper,twocolumn,11pt,accepted=2026-07-02]{quantumarticle}
\pdfoutput=1
\usepackage[utf8]{inputenc}
\usepackage[T1]{fontenc}
\usepackage[english]{babel}
\usepackage{amsmath,amssymb,amsthm}
\usepackage{dsfont}
\usepackage[braket,qm]{qcircuit}
\usepackage{graphicx}
\usepackage{xcolor}
\usepackage[numbers,sort&compress]{natbib}
\usepackage[
    colorlinks=true,
    linkcolor=blue,
    citecolor=blue,
    urlcolor=blue
]{hyperref}
\newtheorem{theorem}{Theorem}
\newtheorem{lemma}[theorem]{Lemma}
\theoremstyle{definition}
\newtheorem{definition}{Definition}
\makeatletter
\newtheorem*{rep@theorem}{\rep@title}
\newcommand{\newreptheorem}[2]{%
\newenvironment{rep#1}[1]{%
 \def\rep@title{#2 \ref{##1}}%
 \begin{rep@theorem}}%
 {\end{rep@theorem}}}
\makeatother

\newreptheorem{theorem}{Theorem}

\begin{document}

\title{Fast quantum measurement tomography with optimal error bounds}

\author{Leonardo Zambrano}
\email{leonardo.zambrano@icfo.eu}
\affiliation{ICFO -- Institut de Ciències Fotòniques, The Barcelona Institute of Science and Technology, 08860 Castelldefels, Barcelona, Spain}

\author{Sergi Ramos-Calderer}

\affiliation{Centre for Quantum Technologies, National University of Singapore, Singapore.}
\affiliation{Quantum Research Center, Technology Innovation Institute, Abu Dhabi, UAE}

\author{Richard Kueng}
\affiliation{Department of Quantum Information and Computation at Kepler (QUICK), Johannes Kepler University Linz, 4040 Linz, Austria}

% \date{\today}

\begin{abstract}
We present a two-step protocol for quantum measurement tomography that is light on classical co-processing cost and still achieves optimal sample complexity. Given measurement data from a known probe state ensemble, we first apply least-squares estimation to produce an unconstrained approximation of the POVM, and then project this estimate onto the set of valid quantum measurements. For a POVM with $L$ outcomes acting on a $d$-dimensional system, we show that the protocol requires $\mathcal{O}\left((d^3+d^2L)/\epsilon^2\right)$ samples to achieve error $\epsilon$ in worst-case distance, and $\mathcal{O}(d^2 L/\epsilon^2)$ samples in average-case distance. We further establish two matching sample complexity lower bounds of $\Omega((d^3 + d^2 L) /\epsilon^2)$ and $\Omega(d^2 L/\epsilon^2)$ for any non-adaptive, single-copy POVM tomography protocol. Hence, our projected least squares POVM tomography is sample-optimal in both the dimension and the number of outcomes for both distances.
Our method admits an analytic form when using global or local 2-designs as probe ensembles and enables rigorous non-asymptotic error guarantees. Finally, we also complement our findings with empirical performance studies carried out on a noisy superconducting quantum computer with flux-tunable transmon qubits.
\end{abstract}
\maketitle

\section{Introduction}
% Motivation

\begin{figure*}[t!]
    \centering
    \includegraphics[width=0.90\linewidth]{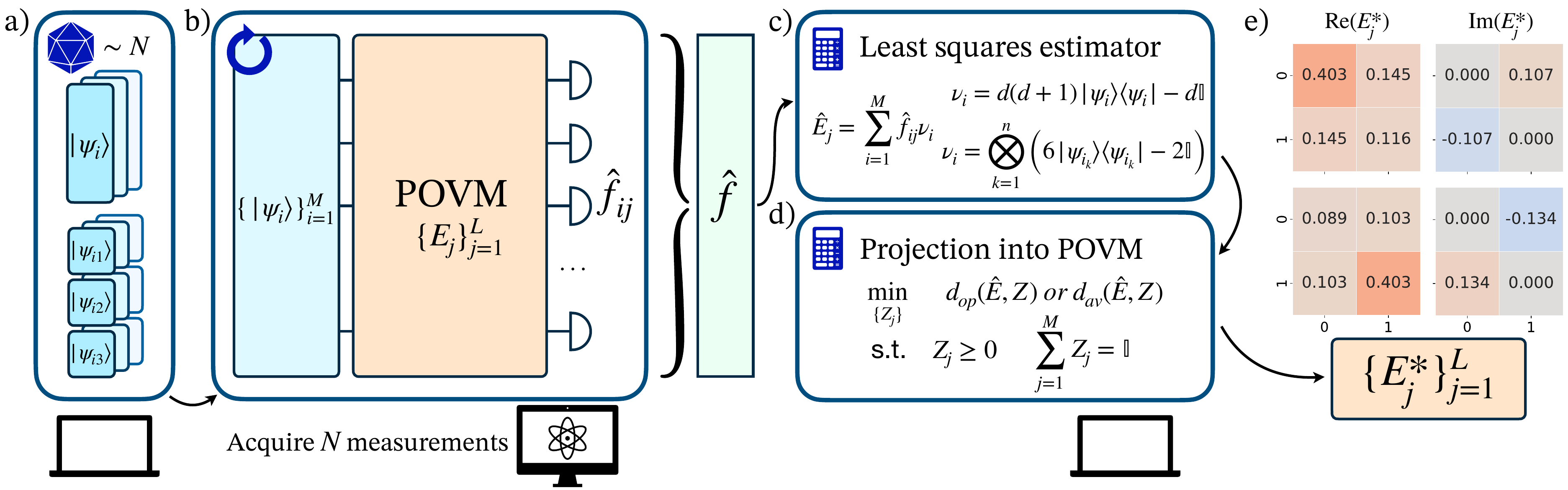}
    \caption{Sketch of the quantum measurement tomography protocol. a) Uniformly sample a set of input states from an IC ensemble $\{\ket{\psi_i}\}$ b) A quantum device is initialized in each randomized input state and measured using the POVM. The resulting frequencies $\hat{f}_{ij}$ for input state $\ket{\psi_i}$ and POVM outcome $E_j$ are stored in a vector $\hat{f}$. c) A least-squares estimator for the POVM is computed using the measured frequencies $\hat{f}$ and the operators $\nu_i$. This step depends on the choice of ensemble, such as a global or local 2-design. d) The raw least squares estimate is projected into the physical subspace of all POVMs %SDPs 
    via convex optimization. This produces the reconstructed POVM $(E_1^*,\ldots,E_L^*)$ as output. e) Example of the first two elements of a single qubit Symmetric Informationally Complete (SIC-)POVM reconstructed from an experimental implementation on two flux-tunable transmon qubits.}
    \label{fig:fig1}
\end{figure*}

A significant limitation to the practical utility of quantum computers is the presence of errors and decoherence. Without active mitigation or error correction, these effects quickly accumulate and destroy the fragile quantum coherence needed for reliable computation~\cite{temme2017error, preskill2018quantum, endo2021hybrid}. Among the different noise sources, measurement errors play a particularly important role, especially in superconducting qubit platforms where readout is often noisier than gate operations~\cite{kim2023evidence}. Reliable characterization of detector noise is therefore crucial, not only for benchmarking hardware but also for enabling error-mitigation strategies.

Accurate detector characterization has already proven useful in several error-mitigation strategies. One prominent approach is to use the information about the detector to correct outcome probabilities or expectation values of observables, thereby reducing the impact of readout noise on subsequent experiments~\cite{li2017efficient, temme2017error, chen2019detector, maciejewski2020mitigation, bravyi2021mitigating, van2022model, funcke2022measurement, korhonen2025practical}. Detector characterization is also important in more general estimation frameworks, such as classical shadows, where the ideal measurement channel assumed in the protocol can be replaced by a noisy one obtained from the characterization~\cite{huang2020classicalshadow, chen2021robustshadow, koh2022classical, brieger2023cgst}. In both cases, reliable knowledge of the detectors enables more accurate data processing and extends the practical capabilities of noisy quantum devices.

Quantum measurement tomography (QMT), also known as detector tomography, provides the standard framework for detector characterization by reconstructing the unknown measurement, modeled as a positive operator-valued measure (POVM), from experimental data. Several approaches to QMT have been proposed and implemented over the years~\cite{fiuravsek2001maximum, d2004quantum, lundeen2009tomography, feito2009measuring, zhang2012recursive, cattaneo2023self, grandi2017experimental, wang2021two}. These methods collect statistics from a set of known probe states and then infer a POVM via classical co-processing. A common strategy is to fit the POVM by solving a constrained optimization problem that minimizes the discrepancy between the observed frequencies and those predicted by the estimated POVM~\cite{fiuravsek2001maximum, d2004quantum, lundeen2009tomography, feito2009measuring, zhang2012recursive, cattaneo2023self}. While such techniques are widely used and often accurate, they become computationally expensive as the system dimension grows. Additionally, existing convergence guarantees for the estimated POVM are only valid in the asymptotic regime, where the central limit theorem applies. Note that this means that for finite sample sizes, error bars are not necessarily grounded in valid theory. Least-squares estimation (LSE) combined with projection methods offers a less computationally demanding alternative~\cite{grandi2017experimental, wang2021two, xiao2023regularization, nielsen2021gate}, and allows for a simpler derivation of reconstruction error bounds. Importantly, these error bounds also hold in the non-asymptotic setting.

Determining the sample complexity of QMT, that is, the number of POVM uses required to reconstruct a measurement with a given accuracy, is a fundamental challenge. Similar to quantum state tomography \cite{o2016efficient, haah2017sample, kueng2017low, guta2020fast, surawy2022projected, chen2022tight, lowe2022lower, anshu2024survey, zambrano2026quantum}, such analysis enables the development of protocols that minimize the use of physical resources. This is particularly important in scenarios where time, memory, or access to quantum devices is limited. Therefore, studying sample complexity supports the design of optimized protocols, supporting practical applications in quantum computing and the experimental verification of quantum devices.

Given the rapid advances in controllable quantum architectures, it is desirable to develop methods that are optimal in terms of sample complexity and also efficient in classical co-processing cost (think: memory and runtime). 
In this paper, we propose a two-step protocol that separates the estimation and constraint satisfaction processes within QMT to obtain an overall more computationally efficient reconstruction protocol. 
Following data acquisition, we use LSE to obtain an initial estimate of a POVM.  LSE provides a computationally simple method to extract matrices that approximate the POVM, but without enforcing the required physical constraints. In a second step, we project this estimate onto the set of valid POVMs, ensuring that the final reconstructed operator satisfies all necessary conditions.

Our main contribution lies in analyzing the sample complexity of this two-step protocol and proving that its scaling matches fundamental lower bounds from quantum information theory. More precisely, for a POVM on a $d$-dimensional system with $L$ outcomes, we show that our protocol requires $\mathcal{O} \left({(d^3 + d^2 L )}/{\epsilon^2} \right)$ samples to achieve error $\epsilon$ in a worst-case distance, and $\mathcal{O} \left({d^2 L }/{\epsilon^2} \right)$ in an average-case distance. Conversely, we prove that \emph{any} protocol based on independent, non-adaptive measurements must use the POVM at least $\Omega \left({(d^3 + d^2 L)}/{\epsilon^2}\right)$ times to achieve the same error in worst-case distance, and $\Omega \left({d^2 L}/{\epsilon^2}\right)$ in average-case distance. Thus, our protocol achieves optimal sample complexity for both distances.

To evaluate the reconstruction error, we consider two natural distance measures between POVMs. The \emph{operational distance} quantifies the worst-case distinguishability of measurement outcomes and is commonly used in discrimination and error-mitigation tasks~\cite{navascues2014how, puchala2018strategies, maciejewski2020mitigation, maciejewski2023exploring, zambrano2026certification}. In contrast, the \emph{average-case distance} captures the expected total variation distance between output distributions when input states are drawn from a random ensemble~\cite{maciejewski2023exploring}. This measure reflects the typical performance of the POVM in practical scenarios. Importantly, having rigorous guarantees on how these errors scale with the number of samples provides a solid theoretical foundation for error-mitigation strategies~\cite{li2017efficient, temme2017error, chen2019detector, maciejewski2020mitigation, bravyi2021mitigating, van2022model, funcke2022measurement, korhonen2025practical, chen2021robustshadow, koh2022classical, brieger2023cgst}, which often assume precise characterization of the measurement noise without explicitly quantifying the estimation error introduced by finite sample sizes.  {We make this explicit by showing how our tomographic bounds directly control the residual error of a classical noise–mitigation scheme~\cite{maciejewski2020mitigation}.} 

Furthermore, we investigate the practical scalability of our method through numerical simulations, benchmarking it against standard Maximum Likelihood Estimation (MLE) measurement tomography. We show that using an alternating projection algorithm \cite{barbera2025boosting}, our protocol matches the reconstruction accuracy of MLE while maintaining computational efficiency for system sizes where standard optimization solvers become intractable.

Finally, we demonstrate the performance of our protocol on a noisy quantum system, successfully recovering experimentally implemented POVMs on two flux-tunable transmon qubits. By establishing the sample-optimality of our method and validating it in real experimental settings, we provide both rigorous theoretical guarantees and a practical tool for quantum measurement characterization.

\section{Preliminaries}\label{sec:pre}

In quantum mechanics, measurements are commonly described using the framework of positive operator-valued measures, or POVMs for short. For a quantum system of dimension $d$, a POVM $E$ is a collection of linear operators $\{ E_j \}_{j=1}^L$ acting on $\mathbb{C}^d$ that satisfy the conditions
\begin{align}\label{eq:conditions_povm} 
\sum_{j=1}^L E_j = \mathds{1} \qquad \textrm{and} \qquad E_j \geq 0
\end{align} 
for every $j = 1, 2, ..., L$. Here, $E_j \geq 0$ requires that each operator is positive semidefinite, i.e. $E_j^\dagger=E_j$ (Hermitian) and every eigenvalue is nonnegative. 
When a quantum state $\rho$ is measured using $E$, the probability of outcome $j$ is given by the Born rule $p_j = \text{tr}(\rho E_j)$. The conditions in Eq.~\eqref{eq:conditions_povm} ensure that this is a valid probability distribution, regardless of the quantum state $\rho$. 
Note that these probabilities depend on both the state and the POVM in question. Hence, they can be used to perform quantum measurement tomography. In this task, the goal is to reconstruct an unknown POVM using the measurement statistics obtained by applying it to a set of known quantum states $\{ \rho_i \}_{i=1}^N$.

To assess the accuracy of tomography, we need a way to quantify the difference between POVMs. One natural choice is the operational distance $d_{\mathrm{op}}$, defined as the largest total variation distance between the probability distributions that two POVMs can produce when applied to the same quantum state \cite{navascues2014how, puchala2018strategies, maciejewski2020mitigation, maciejewski2023exploring}.
\begin{definition}[Operational distance]
    Let $E$ and $F$ be two $L$-outcome POVMs acting on a Hilbert space of dimension $d$. We define the operational distance $d_{\mathrm{op}}$ as 
\begin{align}
    d_{\mathrm{op}} (E, F) = \max_{\rho \, \in \mathcal{D}(\mathcal{H})} \frac{1}{2} \sum_{j=1}^L | \mathrm{tr}(E_j \rho) -  \mathrm{tr}(F_j \rho) |. \label{eq:dop_l1}
\end{align}
\end{definition}
This is equivalent to
\begin{align}
    d_{\mathrm{op}} (E, F) = \max_{x \in \mathcal{P}(L)} \left\Vert \sum_{k \in x} (E_k - F_k) \right\Vert, \label{eq:dop_linf}
\end{align}
where $\mathcal{P}(L)$ is the power set of $\{1, 2, ..., L\}$ and $\Vert A \Vert$ is the spectral norm of $A$, i.e.\ the largest matrix eigenvalue in modulus. The number $d_\mathrm{op}(E, F) \in [0,1]$ is related to the optimal probability of distinguishing between measurements $E$ and $F$ without using entanglement with a larger auxiliary system. 

In our proofs, we will use
\begin{align}\label{eq:extended_op_norm_main}
    d_{\mathrm{ext}}(X,Y)
    =
    \frac12
    \sup_{a \in \{\pm 1\}^{L}}
    \left\|
        \sum_{j=1}^L a_j (X_j-Y_j)
    \right\|.
\end{align}
where $X=\{X_j\}_{j=1}^{L}$ and $Y=\{Y_j\}_{j=1}^{L}$ for $X_j$ and  $Y_j$ Hermitian. If $E$ and $F$ are POVMs, then $d_{\mathrm{ext}}(E, F)=d_{\mathrm{op}} (E,F)$.

The operational distance is a worst-case distance measure. In practical scenarios, however, it is often more relevant to compare the average-case behavior of quantum measurements. To capture this, we can consider the expected Total Variation (TV) distance between the measurement statistics of two POVMs, $E$ and $F$, when applied to pure states $|\psi\rangle$ drawn from a random ensemble $\nu$
\begin{align}
    \mathrm{TV}_{\mathrm{av}}(E, F) = \mathop{\mathbb{E}}_{|\psi\rangle} \left[ \frac{1}{2} \sum_{i=1}^L \left| \langle \psi | E_i | \psi \rangle - \langle \psi | F_i | \psi \rangle \right| \right]. 
\end{align}
While calculating this expected TV distance directly is computationally challenging, if the ensemble $\nu$ forms at least an approximate state $4$-design, $\mathrm{TV}_{\mathrm{av}}(E, F)$ is proportional to the following compact expression \cite[Theorem~2]{maciejewski2023exploring}:

\begin{definition}[Average case distance]
    Let $E$ and $F$ be two $L$-outcome POVMs acting on a Hilbert space of dimension $d$. We define the average-case distance $d_{\mathrm{av}}$ as
\begin{align}
    d_\mathrm{av} (E, F) = \frac{1}{2d} \sum_{i=1}^L \sqrt{\Vert E_i - F_i \Vert_F^2 + \left( \mathrm{tr}[E_i - F_i] \right)^2},
\end{align}
where $\Vert A \Vert_F = \sqrt{\mathrm{tr}[A^\dagger A]}$ denotes the Frobenius norm.
\end{definition}
This distance is related to the operational distance by $a\,d_{\mathrm{av}}(E,F) \le d_{\mathrm{op}}(E,F) \le d\,d_{\mathrm{av}}(E,F)$, where $a=0.31$~\cite{maciejewski2023exploring}.

Finally, we base our unconditional lower bounds on the sample complexity required for any quantum measurement protocol on Fano's inequality. More precisely, we use the following consequence which is displayed in \cite[Corollary~2.7]{lowe2022lower}.

\begin{lemma}\label{thm:fano}
    Let $X$, $Y$, and $\hat{X}$ be discrete random variables forming a Markov chain $X \rightarrow Y \rightarrow \hat{X}$, where $X$ takes values in $\mathcal{X}$. Suppose Alice sends a message $X \sim \mathrm{Unif}(\mathcal{X})$ and Bob is able to decode the message with constant probability of success using $\hat{X}$. Then, it must hold that
\begin{align}
    I(X : Y) = \Omega(\log |\mathcal{X}|),
\end{align}
where $I$ denotes the mutual information
\end{lemma}

\section{Projected least-squares quantum measurement tomography}

Here, we present a protocol to reconstruct any $L$-outcome POVM on a $d$-dimensional system with a sample complexity of $N \in \mathcal{O} \left(\frac{(d^3 + d^2 L)}{\epsilon^2} \right)$ in the operational distance $d_{\mathrm{op}}$ (worst case), and $N \in \mathcal{O} \left(\frac{d^2 L}{\epsilon^2}\right)$ in the average distance $d_{\mathrm{av}}$. A sketch of this protocol, and a sample of reconstructed empirical data are shown in Fig. \ref{fig:fig1}.  

\subsection{Protocol}
Our aim is to characterize a POVM $E$ with $L$ elements  $\left\{ E_1, E_2, \dots, E_L \right\}$ acting on a Hilbert space of dimension $d$. To this end, we collect measurement statistics from a known set of quantum states $\{|\psi_i\rangle\}_{i=1}^M$ . This set is chosen to be informationally complete, meaning that for any POVM element $E_j$, the probabilities $\langle \psi_i |E_j|\psi_i\rangle$ uniquely identify $E_j$. We then perform least-squares estimation to obtain preliminary estimators for the POVM elements, which are subsequently projected onto the set of physical POVMs.

It is worthwhile to point out that the least-squares estimation step can be computed analytically if the ensemble of probe states $\{|\psi_i\rangle\}_{i=1}^M$ either forms a global $2$-design or is a tensor product of single-qubit $2$-designs. This $2$-design probe ensemble is distinct from the 4-design ensemble used to  motivate $d_{\mathrm{av}}$ in Sec.~\ref{sec:pre}. While relating $d_{\mathrm{av}}$ to the average TV distance requires a 4-design, the tomographic reconstruction itself only requires a 2-design to analytically invert the measurement map and bound the estimator's variance. If the ensemble is a tensor product of single-qubit $2$-designs, each state is of the form $|\psi_i \rangle = |\psi_{i_1} \rangle \otimes |\psi_{i_2} \rangle \otimes \dots \otimes |\psi_{i_n} \rangle$, where the states $\{|\psi_{i_k} \rangle \}$ of each qubit $k = 1, \dots, n$ form a $2$-design. Examples of $2$-designs include stabilizer states \cite{dankert2009exact}, states from a symmetric informationally complete POVM (SIC-POVM) \cite{renes2004symmetric}, and states forming a maximal set of mutually unbiased bases (MUBs) \cite{klappenecker2005mutually}.
We propose the following estimation protocol:

\begin{enumerate}
    \item \textit{Data collection}. We uniformly sample a state from the set $\{| \psi_i \rangle \}_{i=1}^{M}$, prepare it on the quantum device and then measure it using the POVM. This procedure is repeated $N$ times. Let $N_{ij}$ denote the number of shots in which the sampled input state was
    $|\psi_i\rangle$ and the observed POVM outcome was $j$. We define $\hat f_{ij}=N_{ij}/N$. Then $\hat f_{ij}$ is an unbiased estimator of the joint probability
    \begin{align}
        p_{ij}
        =
        \frac{1}{M}\langle\psi_i|E_j|\psi_i\rangle.
    \end{align}
    \item \textit{Least squares estimator}. We compute empirical estimators $\hat{E}_j$ of each POVM element separately as
    \begin{align}\label{eq:lse_estimator}
        \hat{E}_j = \sum_{i=1}^M \hat{f}_{ij} \nu_i, 
    \end{align}
    where 
    \begin{align}\label{eq:global_frame}
        \nu_i =  d(d+1) |\psi_i \rangle \langle \psi_i | - d \mathds{1}
    \end{align}
    if the set of states $\{|\psi_i \rangle\}_{i=1}^M$ forms a $2$-design (see App.~\ref{app:2-designs}) or
    \begin{align}\label{eq:local_frame}
        \nu_i = \bigotimes_{k=1}^n\left(6 | \psi_{i_k} \rangle \langle \psi_{i_k} | -  2 \mathds{1}\right)
    \end{align}
    if the set of states is a tensor product of single-qubit $2$-designs (see App.~\ref{app:single_qubit-2-designs}). The elements from Eq.~$\eqref{eq:lse_estimator}$ form an unbiased estimator for the POVM, denoted as $\hat{E}$.
    \item \textit{Projection onto a POVM}. The bare least squares estimator $\hat{E} = \left\{\hat{E}_1,\ldots,\hat{E}_L \right\}$ will typically not meet the conditions of a physical POVM.  
    To obtain a physical object we solve the following convex optimization problem  \cite{vandenberghe1996semidefinite}
    \begin{align}\label{eq:proj_povm}
    \underset{ \{Z_1,\ldots,Z_L\} \subset \mathbb{C}^{d \times d}}{\text{minimize}}  \; & \quad d_{\mathrm{ext}}(\hat{E}, Z) \: \textrm{ or } \: d_{\mathrm{av}} (\hat{E}, Z) \nonumber\\
    \textrm{subject to}  & \quad Z_j^\dagger = Z_j, \quad Z_j \geq 0\\
    & \quad \sum_{j=1}^L Z_j = \mathds{1}. \nonumber
    \end{align}
    We output the optimal solution (argmin) $\{E_j^*\}_{j=1}^L$ of this optimization problem as the estimator of the POVM in the corresponding distance. 
\end{enumerate}

The LSE in Eqs. \eqref{eq:global_frame} and \eqref{eq:local_frame} are related to single-shot estimators in quantum state tomography and classical shadows \cite{guta2020fast, huang2020classicalshadow}. Due to the symmetry of the Born rule, probing an unknown POVM with random states is the dual of probing an unknown state with random measurements. Consequently, our estimators $\nu_i$ are exactly $d$ times the shadow estimators. This factor of $d$ accounts for the differing trace constraint, $\mathrm{tr}(\rho) = 1$ versus $\sum_{j=1}^L \mathrm{tr}(E_j) = d$.

\subsection{Error analysis for operational distance $d_{\mathrm{op}}$}
  Now, we will show that the worst-case error $d_{\mathrm{op}}(E, E^*)$ can be controlled by the sample size $N$.

\begin{theorem}\label{thm:sc_dop}
    Assume an $L$-outcome POVM $E$ acting on a Hilbert space of dimension $d$, and an IC set of states $\{|\psi_i\rangle\}_{i=1}^{M}$. A reconstruction error $d_{\mathrm{op}}(E, E^*) \leq \epsilon$ with probability $1 - \delta$ can be achieved using the protocol described above with sample size
    \begin{align}\label{eq:N_global_op_main}
        N &\geq \frac{8\left(d^2+\epsilon(d^2+1)/6\right)}{\epsilon^2} \ln \left(\frac{2^{L+1}9^{2d}}{\delta}\right)
    \end{align}
    when $\{|\psi_i\rangle\}_{i=1}^{M}$ forms a $2$-design, or
    \begin{align}\label{eq:N_local_op_main}
        N &\geq  \frac{8\left(6^n+\frac{\epsilon}{6}(4^n+1)\right)}{\epsilon^2}\ln \left(\frac{2^{L+1}9^{2d}}{\delta}\right).
    \end{align}
    if the POVM acts on an $n$-qubit system and $\{|\psi_i\rangle\}_{i=1}^{M}$ is a tensor product of single-qubit $2$-designs.
\end{theorem}

These bounds imply that the sample complexity of our protocol scales as $\mathcal{O}\left( (d^3 + d^2 L) / \epsilon^2 \right)$ when using global 2-designs, and as $\mathcal{O}\left( (12^n + 6^n L)  / \epsilon^2 \right)$ for tensor product of 2-designs. The number of degrees of freedom in an $L$-outcome POVM on a $d$-dimensional Hilbert space is $\Omega(d^2 L)$. Thus the bound matches the parameter-count scaling whenever $L\gtrsim d$, while for few-outcome measurements it contains an additional $d^3$ term.

\begin{proof}
We first control the error in least-squares estimator. Fix a sign vector $a=(a_1,\ldots,a_L)\in\{-1,+1\}^L$ and define the signed fluctuation
\begin{align}
    A_a = \sum_{j=1}^L a_j(\hat E_j-E_j).
\end{align}
The vector $a$ is fixed throughout the following scalar concentration
argument. Later we take a union bound over all $2^L$ choices of $a$.

Fix a unit vector $|u\rangle\in\mathbb C^d$ and
$a\in\{\pm1\}^L$. Let
\begin{align}
    g_i(u) =\langle u|\nu_i|u\rangle .
\end{align}
For this fixed $u$ and fixed $a$, define the real-valued random variable $Z_{u,a}$ such that
\begin{align}
    Z_{u,a}(i,j) = a_j g_i(u).
\end{align}
with probability $p_{ij} = \frac{1}{M} \langle\psi_i|E_j|\psi_i\rangle$. 

The empirical mean of $Z_{u,a}$ is
\begin{align}
    \sum_{i=1}^M\sum_{j=1}^L
    \hat f_{ij} Z_{u,a}(i,j)
    &=
    \sum_{i=1}^M\sum_{j=1}^L
    \hat f_{ij}a_j\langle u|\nu_i|u\rangle
    \nonumber \\
    &=
    \langle u |\left(
    \sum_{j=1}^L a_j\hat E_j
    \right) |u\rangle .
\end{align}

The expectation of $Z_{u,a}$ is
\begin{align}
    \mathbb E [Z_{u,a}]
    &=
    \sum_{i=1}^M\sum_{j=1}^L
    p_{ij}a_j\langle u|\nu_i|u\rangle
    \nonumber \\
    &=
    \sum_{j=1}^L a_j
    \langle u| \left(
    \frac{1}{M}\sum_{i=1}^M
    \langle\psi_i|E_j |\psi_i \rangle \nu_i \right)
     |u \rangle \nonumber\\
     & =  \langle u| \left(
    \sum_{j=1}^L a_jE_j
    \right)|u\rangle.
\end{align}
% add justification from the least squares estimator. 
Hence, the difference between these two quantities is
\begin{align}\label{eq:quad_form_freqs_clean}
    \langle u|A_a|u\rangle
    =
    \sum_{i=1}^M\sum_{j=1}^L
    (\hat f_{ij}-p_{ij})a_jg_i(u).
\end{align}

By scalar Bernstein's inequality, for every fixed unit vector $u$ and every
fixed sign vector $a$,
\begin{align}\label{eq:scalar_bernstein_main_clean}
    \Pr\left(
    |\langle u|A_a|u\rangle|\geq t
    \right)
    \leq
    2\exp\left(
    -
    \frac{N t^2}
    {2(v^2+Kt/3)}
    \right).
\end{align}
The Bernstein parameters for $Z_{u,a}$ are computed in
Appendices~\ref{app:global_2design_scalar_bernstein} and
\ref{app:local_2design_scalar_bernstein}. Namely,
\begin{align}
        \operatorname{Var}(Z_{u,a})\le v^2,
    \qquad
    |Z_{u,a}-\mathbb E Z_{u,a}|\le K,
\end{align}
with $(v^2,K)=(d^2,d^2+1)$ for a global $2$-design and $(v^2,K)=(6^n,4^n+1)$ for a tensor product of single-qubit $2$-designs.

We now pass to operator norm. Let $\mathcal N$ be
a $1/4$-net of the unit sphere of $\mathbb C^d$. We choose it such that $|\mathcal N| \leq 9^{2d}$ \cite{vershynin_2018}. For every Hermitian matrix $B$, the standard net estimate gives \cite{vershynin_2018}
\begin{align}\label{eq:net_operator_norm_main_clean}
    \|B\|
    \leq
    2\max_{u\in\mathcal N}
    |\langle u|B|u\rangle|.
\end{align}

Apply Eq.~\eqref{eq:scalar_bernstein_main_clean} with $t=\epsilon/2$, and
take a union bound over all $u\in\mathcal N$ and all
$a\in\{-1,+1\}^L$. The failure probability is at most
\begin{align}
    2^{L+1}9^{2d}
    \exp\left(
    -
    \frac{N\epsilon^2}
    {8(v^2+K\epsilon/6)}
    \right).
\end{align}
The sample-size assumptions in Eqs.~\eqref{eq:N_global_op_main} and
\eqref{eq:N_local_op_main} make this quantity at most $\delta$. Hence, with
probability at least $1-\delta$,
\begin{align}
    \max_{u\in\mathcal N}
    |\langle u|A_a|u\rangle|
    \leq
    \frac\epsilon2
\end{align}
simultaneously for every sign vector $a\in\{-1,+1\}^L$. By
Eq.~\eqref{eq:net_operator_norm_main_clean}, this implies
\begin{align}
    d_{\mathrm{ext}}(E, \hat{E})
    =
    \frac{1}{2}\max_{a\in\{-1,+1\}^L}
    \left\|
    \sum_{j=1}^L a_j(\hat E_j-E_j)
    \right\| 
    \leq
    \frac{\epsilon}{2}.
\end{align}

It remains to use the projection step. Since $E$ is itself a feasible POVM,
the definition of $E^*$ gives $d_{\mathrm{ext}}(E^*, \hat{E})
    \leq
    d_{\mathrm{ext}}(E, \hat{E})$. Using the triangle inequality,
\begin{align}
    d_{\mathrm{ext}}(E, E^*)
    &\leq
    d_{\mathrm{ext}}(E^*, \hat{E})
    +
    d_{\mathrm{ext}}(\hat{E}, E )
    \nonumber\\
    &\leq
    2 d_{\mathrm{ext}}(\hat{E}, E )
    \leq
    \epsilon.
\end{align}
Finally, $E$ and $E^*$ are POVMs. Therefore the extended operational norm
coincides with the operational distance:
\begin{align}
    d_\mathrm{op}(E,E^*)
    =
    d_{\mathrm{ext}}(E, E^*)
    \leq
    \epsilon.
\end{align}
This proves the theorem.
\end{proof}

Projecting $\hat{E}$ onto the set of POVMs using $d_{\mathrm{ext}}$ may become computationally demanding when $L$ is large. As a practical alternative, one can perform the projection using any distance that upper bounds $d_{\mathrm{ext}}$ \cite{wang2021two, barbera2025boosting}. A natural choice is $d_\infty = \frac{1}{2}\min_{Z \in \mathrm{ POVM} } \sum_{i=1}^L \Vert \hat{E}_i -  Z_i \Vert$. Then, 
\begin{align}
    d_{\mathrm{op}}(E, E_\infty^*) &\leq d_{\mathrm{ext}}(E, \hat{E}) + d_{\mathrm{ext}}(E_\infty^*, \hat{E})\\
    &\leq  \frac{\epsilon}{2} + d_\infty. 
\end{align}
$d_\infty$ can be calculated after the protocol is run, and in practice is of the order of $\epsilon$. 
\subsection{Error analysis for average distance $d_{\mathrm{av}}$}
  A similar calculation can be done to bound the error in the average distance $d_{av}$:

\begin{theorem}\label{thm:upper_bound_d_av}
    Assume an $L$-outcome POVM $E$ acting on a Hilbert space of dimension $d$ and an IC set of states $\{|\psi_i \rangle \}_{i=1}^{M}$. A reconstruction error $d_{\mathrm{av}}(E, E^*) \leq \epsilon$ with probability $1-\delta$ can be achieved using the protocol described above with sample size
    \begin{align}
        N &\geq \frac{ 8L\left(d^2 + d + {d\epsilon}/{3\sqrt{L}} \right)}{ \epsilon^2}  \ln \left(\frac{4}{\delta}\right)
    \end{align}
    when $\{|\psi_i\rangle \}_{i=1}^{M}$ forms a $2$-design, or
    \begin{align}
        N \geq \frac{8L\left(5^n + {5^{n/2}\epsilon}/{3\sqrt{L}}\right)}{\epsilon^2}  \ln \left(\frac{4}{\delta}\right)
    \end{align}
    if the POVM acts on an $n$-qubit system and $\{|\psi_i\rangle \}_{i=1}^{M}$ is a tensor product of single-qubit $2$-designs.
\end{theorem}

These sample complexity guarantees scale as order $\mathcal{O}(d^2 L / \epsilon^2)$ in the global 2-design setting, and as $\mathcal{O}(5^n L / \epsilon^2)$ when using local 2-designs. Compared to the worst-case (operational) distance, this represents a reduction in sample complexity by a factor of $d$.  The detailed proof is given in App.~\ref{app:proof_thm_4}.

\section{Minimum sample complexity}

Here, we provide a lower bound for the sample complexity of non-adaptive measurement tomography considering the worst-case distance $d_{\mathrm{op}}$. The proof is based on a discretization of the problem which allows us to relate quantum measurement tomography to the problem of discrimination of well-separated POVMs.  We first construct a set of $R \in \exp(\Omega(d^2))$ POVMs on dimension $d$ that are $\epsilon/8$ apart in operational distance from each other. We then encode a random message using this set and decode it using measurement tomography with sufficient precision. From Fano's inequality, this gives us a lower bound $\Omega(d^2)$ for the mutual information between the encoder and decoder. Additionally, we obtain an upper bound $\mathcal{O}(N \epsilon^2 / d)$ for the mutual information between the parties after  {$N$} uses of the POVM. Using these two results, we derive a bound $N \geq \Omega(d^3/\epsilon^2)$ for the sample complexity of any non-adaptive tomographic procedure. These techniques have been previously employed in quantum information theory and are standard tools in the statistical estimation literature \cite{scarlett2019introductory, flammia2012quantum, roth2018recovering, haah2017sample, huang2020classicalshadow, lowe2022lower}.

\subsection{Construction of an $\epsilon$-packing}

  Here, we will demonstrate the existence of a set of $R \in \exp (\Omega (d^2 ))$ POVMs that are at least $\epsilon/8$ apart in operational distance. Each POVM $E_{U_x}$ in this set has $L+2$ elements of the form
\begin{align}\label{eq:POVM}
E_{U_x}^{k} &= \frac{1}{2L} \mathds{1} \quad \textrm{for} \quad k=1, \dots, L \nonumber\\
E_{U_x}^{L+1} &= \frac{(1 +\epsilon)}{4} \mathds{1} - \frac{\epsilon}{2} U_x P U_x^\dagger \nonumber \\
E_{U_x}^{L+2} &= \frac{(1 - \epsilon)}{4} \mathds{1} + \frac{\epsilon}{2} U_x P U_x^\dagger,
\end{align}
where $U_x$ is a unitary operator, $P$ a projector of rank $d/2$ and $0 \leq \epsilon \leq 1/2$.

To construct this packing, we require the following set of unitaries \cite{lowe2022lower}:

\begin{lemma}\label{thm:unitaries}
    There exists a set of $R \in \exp(\Omega (d^2))$ unitaries such that for any $U_i \neq U_j$ in the set, 
    \begin{align}
        \frac{1}{d} \Vert U_i P U_i^\dagger - U_j P U_j^\dagger \Vert_1 \geq \frac{1}{4},  
        \end{align}
    where $P$ is a fixed projector of rank $d/2$.
\end{lemma}

We will use this to prove the following

\begin{lemma}\label{thm:packing_dop}
    There exists a set of $R \in \exp(\Omega (d^2))$ POVMs that are $\epsilon/8$ apart in operational distance $d_{\mathrm{op}}$ from each other.
\end{lemma}
\begin{proof}
From Lemma \ref{thm:unitaries} we fix a set of $R \in \exp(\Omega (d^2))$ unitaries to construct a set of POVMs in $\exp(\Omega(d^2))$ of the form of Eq.~\eqref{eq:POVM}.  Then all the POVMs are at least $\epsilon/8$ apart from each other:
\begin{align}
    d_{\mathrm{op}} (E_{U_i}, E_{U_j}) = & \frac{\epsilon}{2} \left\Vert  U_i P U_i^\dagger - U_j P U_j^\dagger \right\Vert \nonumber  \\
        \geq & \frac{\epsilon}{2d} \left\Vert  U_i P U_i^\dagger - U_j P U_j^\dagger \right\Vert_1 \nonumber \\
        \geq & \frac{\epsilon}{8}. 
\end{align}
\end{proof}

In the proof of the lower bound on the sample complexity, each POVM $E_{U_x}$ of the form of Eq.~\eqref{eq:POVM} will encode a uniform random message $X$ that will be decoded using the information given by the measurement outcomes $Y$.

\subsection{Upper bound on mutual information}

To prove the upper bound for the mutual information, we will first bound it in terms of the $\chi^2$-divergence \cite{lowe2022lower}:

\begin{lemma}\label{thm:upper_bound_mi_divergence}
    Let $X$ be an arbitrary random variable and $Z$ a discrete random variable. Then,
    \begin{align}
        I (X:Z) \leq \frac{1}{\ln (2)} \left(\sum_{z} \mathop{\mathbb{E}}_{X \sim p_X} \frac{p_{Z|X} (z)^2}{p_Z(z)} -1\right),
    \end{align}
    where $p_{Z}$ is the probability distribution of $Z$ and $p_{Z|X}$ is the probability distribution of $Z$ conditioned on the event $X = x$. 
\end{lemma}

Measuring a set of states $\{ \rho_i \}_{i=1}^N$ using a random POVM $E_{U_X}$ in our set of POVMs defines a random variable $Y = (Y_1, \dots Y_N)$ such that the probability of outcome $y_i \in [L+2]$ in round $i$ is given by $p (Y_i = y_i ) = \mathrm{tr} \left(E_{U_X}^{y_i} \rho_i  \right)$. 
Computing the mutual information $I(X : Y)$ is challenging because we do not have an explicit description of the POVMs in our construction. The next result, adapted from Proposition 4.3 in Ref.~\cite{lowe2022lower}, allows us to replace our specific ensemble of POVMs with one defined using Haar-random unitaries in the analysis of mutual information:
 \begin{lemma}\label{thm:upper_bound_mutual_information_haar}
     Let $X \sim Unif[R]$, $U$ be a Haar random unitary and $\{\rho_i\}_{i=1}^N$ a set of $N$ quantum states. Then, there exists a set of $R \in \exp(\Omega (d^2))$ POVMs of the form of Eq. \eqref{eq:POVM} which forms an $\epsilon/8$ packing and satisfies 
     \begin{align}
     I(X: Y) \leq I(U: Z), 
     \end{align}
     where $Y = (Y_1, \dots Y_N)$ is the outcome of measuring $\{\rho_i\}_{i=1}^N$ with a random POVM $E_{U_X}$ and $Z= (Z_1, \dots Z_N)$ is the outcome of measuring $\{\rho_i\}_{i=1}^N$ with a random POVM $E_{U}$. 
 \end{lemma}
This result allows us to upper bound the mutual information by averaging over Haar-random unitaries, instead of computing it directly from the specific set of unitaries used in the $\epsilon/8$-packing.

Using these results and following Ref.~\cite{lowe2022lower} we obtain an upper bound for the mutual information:
\begin{lemma}\label{thm:upper_bound_mutual_information_dop}
    Let $X \sim Unif([R])$ and $Y = (Y_1, Y_2, \dots Y_N)$ be the outcome of the measurement of a random POVM $E_{U_X}$ of the form of Eq. \eqref{eq:POVM} over $N$ quantum states. Then, 
\begin{align}
    I(X: Y) \leq \frac{N}{\ln(2)} \frac{\epsilon^2}{(d+1)}
\end{align}
\end{lemma}
\begin{proof}
We have
\begin{align}\label{eq:calculation_upper_bound_mutual_information}
    I(X: Y) \leq & I(U : Z) \nonumber \\
    = & \sum_{i=1}^N I(U:Z_i|Z_{i-1}, \dots Z_{1}) \nonumber\\
    \leq & \sum_{i=1}^N I(U:Z_i)  \nonumber \\
    \leq & \frac{N}{\ln (2)}  \mathop{\mathbb{E}}_{U\sim \mathrm{Haar}} \left(\sum_{z_i} \frac{p_{Z_i|U} (z_i)^2}{p_{Z_i}(z_i)} -1\right), 
\end{align}
where the first line follows from Lemma \ref{thm:upper_bound_mutual_information_haar}, the second one from the chain rule for mutual information, the third one from the independence of $Z_1, \dots Z_N$ given $U$ and the final one from Lemma \ref{thm:upper_bound_mi_divergence}. % and the independence of $U$. 
Then, we have to calculate the values of $p_{Z_i}(z_i)$ and $\mathop{\mathbb{E}}_{U} p_{Z_i|U} (z_i)^2$.
% for  the third one, the independence part, see corollary 4.6 of Lowe and Nayak

For POVM elements $L+1$ and $L+2$, the probability  $p_{Z_i}(z_i) = \mathop{\mathbb{E}}_{U} p_{Z_i  |  U} (z_i)$ is of the form
\begin{align}
    p_{Z_i} (z_i) =&  \mathop{\mathbb{E}}_{U} \mathrm{tr}\left(\left(\frac{(1 \pm \epsilon)}{4}\mathds{1} \mp \frac{\epsilon}{2} U P U^\dagger\right) \rho \right) = \frac{1}{4}.
\end{align}
For $i=1\dots L$ we have $p_{Z_i} (z_i) = 1/(2L)$. Next, we have for POVM elements $L+1$ and $L+2$:
\begin{align}
     \mathop{\mathbb{E}}_{U} p_{Z_i |  U}(z_i)^2 =& \mathop{\mathbb{E}}_{U} \mathrm{tr}\left(\left(\frac{(1 \pm \epsilon)}{4}\mathds{1} \mp  \frac{\epsilon}{2} U P U^\dagger\right) \rho \right)^2  \nonumber\\
     &  {\leq} \frac{1}{16} \left(1+ \frac{ {2} \epsilon^2}{(d+1)} \right)
\end{align}
and for $i=1\dots L$ we have $p_{Z_i} (z_i)^2 = 1/(4L^2)$. Then, collecting these results and substituting in Eq.~\eqref{eq:calculation_upper_bound_mutual_information}, we obtain
\begin{align}
    I(X: Y) &\leq \frac{N}{\ln(2)} \left(\sum_{z_i=1}^L  \frac{1}{2 L} +  \left(\frac{1}{2} + \frac{\epsilon^2}{(d+1)} \right) - 1\right) \nonumber\\
    &=\frac{N}{\ln(2)} \frac{\epsilon^2}{(d+1)}
\end{align}
\end{proof}

\subsection{Lower bound on the sample complexity}
Now we show a lower bound on the sample complexity

\begin{theorem}
    Any procedure for quantum measurement tomography of a POVM on a $d$-dimensional Hilbert space that is $\epsilon/16$ accurate in operational distance using nonadaptive, single-copy measurements on known input states requires
    \begin{align}
        N \in \Omega\left( \frac{d^3 + d^2 L}{\epsilon^2} \right)
    \end{align}
    uses of the unknown POVM. 
\end{theorem}

\begin{proof}
Assume a random message is encoded in $N$ copies of a POVM $E_{U_X}$, where this POVM is uniformly sampled from the $\epsilon/8$-packing defined by Lemma \ref{thm:packing_dop}. Let $Y = (Y_1, \dots, Y_N)$ denote the outcomes from measuring $N$ known quantum states using the POVM. Since each POVM in the packing is separated by at least $\epsilon/8$, a tomography algorithm that takes this data and produces an estimate of the POVM within $\epsilon/16$ precision in $d_{\mathrm{op}}$ (with some constant probability) is sufficient to decode the message. Then, tomography must have a sample complexity at least as large as the communication problem. From Lemma \ref{thm:upper_bound_mutual_information_dop}, the mutual information between the measurements used for tomography and the message is upper bounded by 
\begin{align}
    I(X: Y) &\leq \frac{N}{\ln(2)} \frac{\epsilon^2}{(d+1)}.
\end{align}
Using Lemma \ref{thm:fano}, we also have 
\begin{align}
    I(X : Y) &\geq \Omega(d^2).
\end{align}
Combining these results, we obtain
\begin{align}
    \frac{N}{\ln(2)} \frac{\epsilon^2}{(d+1)} &\geq I(X: Y) \geq \Omega(d^2).
\end{align}
Thus, we conclude that
\begin{align}
    N &\geq \Omega \left( \frac{d^3}{\epsilon^2} \right).
\end{align}

The packing above gives the dimension-dependent contribution
$N=\Omega(d^3/\epsilon^2)$. The outcome-dependent contribution
$N=\Omega(d^2L/\epsilon^2)$ follows from the average-distance packing
in Appendix~\ref{app:minimal_sample_complexity_dav}, together with $0.31\,d_{\mathrm{av}}\le d_{\mathrm{op}}$. Combining the
two independent lower bounds gives
\begin{align}
    N=\Omega\!\left(\frac{d^3+d^2L}{\epsilon^2}\right).
\end{align}
\end{proof}

\begin{figure}[t!]
    \includegraphics[width=0.99\linewidth]{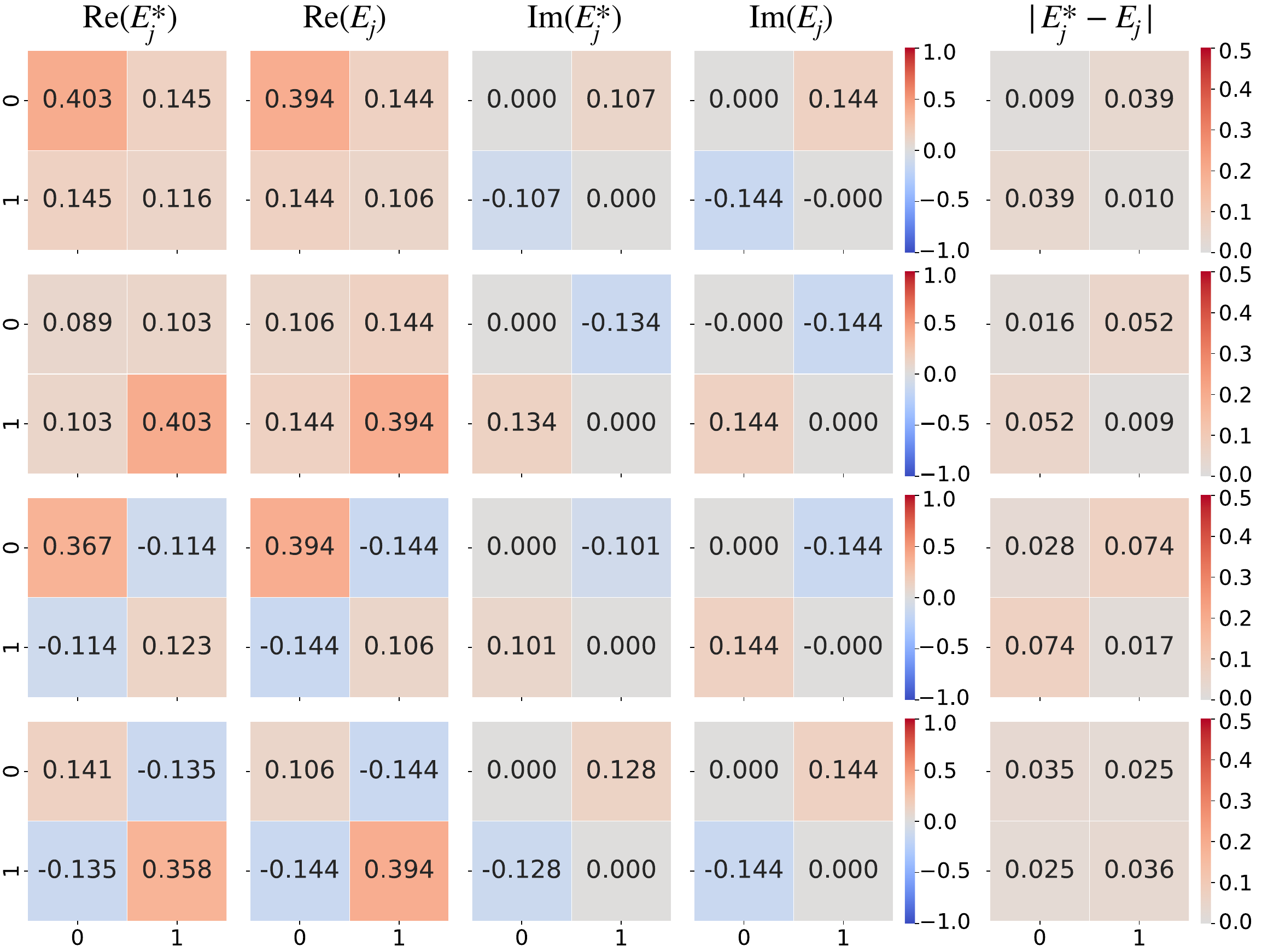}
    \centering
    \caption{Results for the reconstruction of a noisy one-qubit Symmetric Informationally Complete (SIC-)POVM using an auxiliary qubit for the generalized measurement, implemented on a two-qubit flux-tunable transmon device using a budget of $1.66\cdot10^5$ random initial states. The reconstructed noisy POVM is compared to its expected noiseless counterpart, with the absolute difference displayed. Though noisy, the reconstruction follows the target POVM closely. The deviations from the expected values are due to the noisy device implementation.}
    \label{fig:sicpovm}
\end{figure}

\section{Application to error mitigation}\label{sec:application_to_error_mitigation}

In experiments, implemented measurements often deviate from the ideal model due to both classical and non-classical noise. Classical noise can be modeled by a column-stochastic matrix~$\Lambda$ and mitigated by applying~$\Lambda^{-1}$ to the measured probabilities \cite{maciejewski2020mitigation}. Non-classical noise, arising from coherent imperfections in the POVM elements, cannot be corrected by classical post-processing and leads to residual errors. Here we show that combining error mitigation with our tomographic bounds allows one to quantify and correct these errors.

Assume we have an ideal POVM $M_{\mathrm{id}}$, an implemented POVM $M_{\mathrm{exp}}$, and a tomographic estimate $M_{\mathrm{est}}$ of $M_{\mathrm{exp}}$. For any quantum state~$\rho$, we denote the corresponding outcome probability vectors by $p_{\mathrm{id}}(\rho)$, $p_{\mathrm{exp}}(\rho)$, and $p_{\mathrm{est}}(\rho)$. Suppose the estimated POVM admits the decomposition $M_{\mathrm{est}} = \Lambda M_{\mathrm{id}} + \Delta$, where $\Lambda$ is a column-stochastic matrix describing classical outcome mixing, and $\Delta$ collects all remaining (non-classical) imperfections. Our goal is to bound the deviation between the ideal statistics and the classically corrected experimental ones using $\Lambda$, uniformly over all states~$\rho$.

Using the triangle inequality,
\begin{align}
    \|p_{\mathrm{id}} - \Lambda^{-1} p_{\mathrm{exp}}\|_1
    &\le 
        \|p_{\mathrm{id}} - \Lambda^{-1} p_{\mathrm{est}}\|_1 \nonumber \\
    &\quad + \|\Lambda^{-1}(p_{\mathrm{est}} - p_{\mathrm{exp}})\|_1.
\end{align}
Applying the inequality $\|Ax\|_1\le\|A\|_{1\to 1}\|x\|_1$ and taking the maximum over all $\rho$ gives
\begin{align}
    & \max_\rho
    \|p_{\mathrm{id}} - \Lambda^{-1} p_{\mathrm{exp}}\|_1 \nonumber\\
    &\le 
    \|\Lambda^{-1}\|_{1\to 1}
        \max_\rho  \big( \|\Lambda p_{\mathrm{id}} - p_{\mathrm{est}}\|_1
        + \|p_{\mathrm{est}} - p_{\mathrm{exp}}\|_1
    \big).
\end{align}
Then, from the definition of the operational distance, we obtain
\begin{align}
    &\sup_\rho
    \|p_{\mathrm{id}} - \Lambda^{-1} p_{\mathrm{exp}}\|_1 \nonumber \\
    &\le  2\|\Lambda^{-1}\|_{1\to 1}
        \big[
            d_{\mathrm{op}}(\Lambda M_{\mathrm{id}}, M_{\mathrm{est}})
            +
            d_{\mathrm{op}}(M_{\mathrm{est}}, M_{\mathrm{exp}})
        \big].
\end{align}

The first term quantifies how well the classical noise model
$\Lambda M_{\mathrm{id}}$ fits the tomographic reconstruction
$M_{\mathrm{est}}$, and can be calculated classically \cite{maciejewski2020mitigation}, since both terms are known, while the second term represents the tomography error, which can be estimated using the tomographic protocol presented here.

\section{Computational scalability and Numerical Simulations}
\label{sec:numerical_simulations}

To evaluate the computational scalability of the proposed protocol, we benchmark our two-step method against standard MLE. Although our approach separates statistical estimation from constraint enforcement, the exact projection onto the POVM set via semidefinite programming (SDP) can still become computationally demanding. For the operational distance $d_{\mathrm{op}}$, the complexity scales exponentially with the number of outcomes $L$, while for the average-case distance $d_{\mathrm{av}}$ it matches the complexity of MLE. Here, we show that replacing the exact SDP projection with Dykstra's alternating projection method avoids these bottlenecks without loss of reconstruction accuracy.

\subsection{Computational Scalability}
Consider an $L$-outcome POVM acting on a $d$-dimensional Hilbert space, which is described by $\mathcal{O}(d^2L)$ independent real parameters. We compare the computational cost of two reconstruction strategies.

The first is standard MLE with semidefinite constraints \cite{fiuravsek2001maximum}, which maximizes the log-likelihood of the observed frequencies subject to the POVM constraints. Using interior-point methods (for example, CVXPY with the MOSEK solver), the dominant cost arises from assembling and inverting the Hessian at each iteration, yielding a time complexity of $\mathcal{O}(d^6 L^3)$ and a memory complexity of $\mathcal{O}(d^4 L^2)$ \cite{vandenberghe1996semidefinite}.

The second strategy combines LSE with a heuristic projection step in Frobenius norm implemented using Dykstra's algorithm. Instead of solving an exact SDP, this method iteratively projects between the positive semidefinite cone and the sum-to-identity hyperplane \cite{barbera2025boosting}. The dominant operation at each iteration is the eigendecomposition of $L$ matrices of size $d \times d$, resulting in a per-iteration time complexity of $\mathcal{O}(d^3L)$ and space complexity of $\mathcal{O}(d^2 L)$.

\subsection{Numerical Simulations}
\begin{figure*}[t]
    \centering
    % Top Row: Execution Times
    \includegraphics[width=0.9\columnwidth]{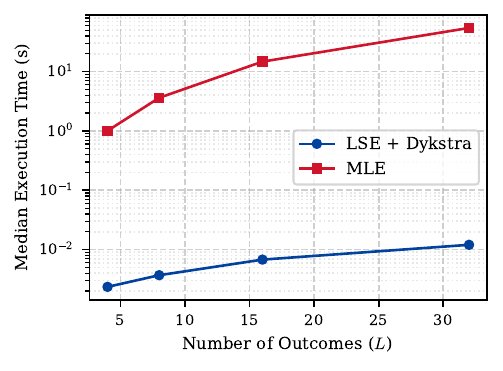}\hfill
    \includegraphics[width=0.9\columnwidth]{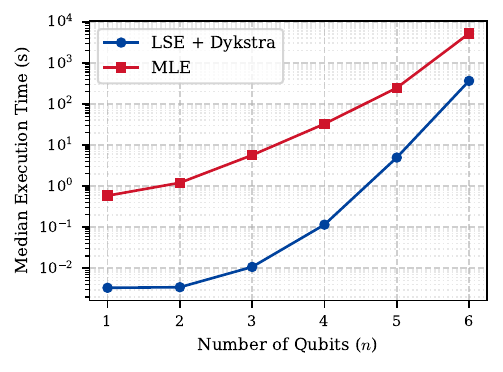} \\
    
    \vspace{-0.1cm} % Adds a little vertical breathing room between rows
    \includegraphics[width=0.9\columnwidth]{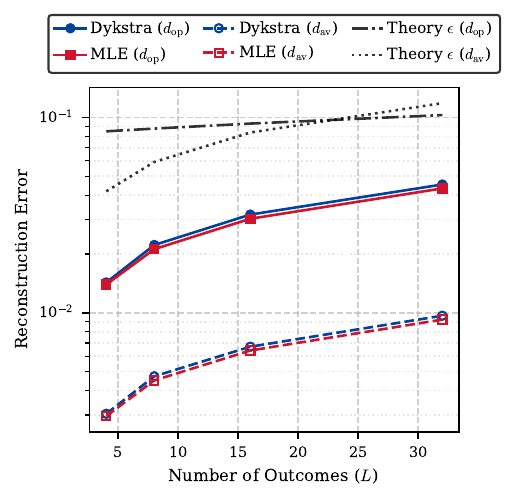}\hfill
    \includegraphics[width=0.9\columnwidth]{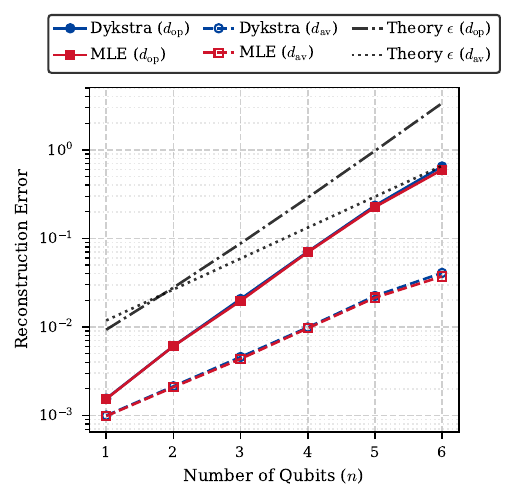}
    
    \caption{Computational efficiency and reconstruction accuracy of the PLS method.
    (Top row) Median classical execution time of our LSE+Dykstra approach compared to standard MLE, scaling with the number of POVM outcomes $L$ for $d=8$ (left), and scaling with the number of qubits $n$ for $L=8$ (right).  
    (Bottom row) Reconstruction errors under the operational distance proxy ($d_{\infty}$, solid lines) and average-case distance ($d_{\mathrm{av}}$, dashed lines) as a function of outcomes $L$ (left) and qubits $n$ (right). Theoretical upper bounds from Thm.~\ref{thm:sc_dop} and Thm.~\ref{thm:upper_bound_d_av} in black. All simulations use $N=10^7$ shots.}
    \label{fig:numerics_master}
\end{figure*}
We evaluate the protocol in two scaling regimes, using a total measurement budget of $N=10^7$ shots. As probe states, we use a local 2-design formed by tensor products of the six single-qubit Pauli eigenstates. First, we fix the system size to $n=3$ qubits and vary the number of outcomes $L \in \{4,8,16,32\}$. Second, we fix the number of outcomes to $L=8$ and scale the system dimension from $n=1$ to $n=6$ qubits.

For each scenario, we generate random target POVMs as follows. We sample $L$ complex matrices $G_j$ with standard normal entries and define positive semidefinite operators $E_j' = G_j G_j^\dagger$. Imposing the completeness relation through $S=\sum_{j=1}^L E_j'$, we construct the POVM elements $E_j = S^{-1/2} E_j' S^{-1/2}$ \cite{wang2021two, barbera2025boosting}. Each shot consists of first sampling a probe state $\rho_i$ uniformly from the
local 2-design and then measuring the unknown POVM. Thus the joint
probabilities are $p_{ij} = \frac{1}{M} \operatorname{tr}(\rho_i E_j)$,
where $M=6^n$ is the number of local 2-design states. The observed counts
$N_{ij}$ are sampled from the multinomial distribution with
probabilities $p_{ij}$, and we define the joint empirical
frequencies $\hat f_{ij}=\frac{N_{ij}}{N}$. Theoretical upper bounds are computed using the analytical expressions derived in Thms.~\ref{thm:sc_dop} and~\ref{thm:upper_bound_d_av} at a confidence level of $1-\delta=0.95$.

Figure~\ref{fig:numerics_master} (top row) displays the median classical runtime over multiple trials as a function of outcomes $L$ and qubits $n$. Standard MLE exhibits steep polynomial scaling, becoming computationally costly for large $L$ and dimension. In contrast, the LSE + Dykstra method is orders of magnitude faster over the same parameter ranges.

Figure~\ref{fig:numerics_master} (bottom row) compares the empirical reconstruction errors of both methods under $d_{\mathrm{op}}$ and  $d_{\mathrm{av}}$. Since evaluating the exact operational distance $d_{\mathrm{op}}$ requires an intractable maximization over all $2^L$ outcome subsets, we instead use the more efficiently computable upper bound $d_\infty = \frac{1}{2} \sum_{j=1}^L \Vert \hat{E}_j - E_j \Vert$. For all evaluated values of $L$ and $n$, Dykstra's algorithm achieves the  same reconstruction quality as full MLE while significantly reducing the computational overhead. The theoretical bounds for both distances are plotted alongside the numerical data, confirming that the observed errors remain consistently below the predicted bounds.

\section{Empirical performance evaluations}\label{experiment}

Here we complement our theoretical findings with empirical performance studies on real quantum hardware.
Target POVMs are prepared using noisy quantum gates and readout. Then, quantum measurement tomography is used to recover the POVM performed on the noisy device. In this section, we present POVMs recovered from a noisy experiment, we showcase the performance of the method in a real scenario and we discuss its potential applications. 

% Brief explanation experimental setup
The experiment is performed on two superconducting flux-tunable transmon qubits. The device is controlled via microwave-based arbitrary waveform generators orchestrated and calibrated by open-source middleware \cite{efthymiou2024qibolab, pasquale2024qibocal}.
For further details, we refer to App.~\ref{app:experimental_setup}.

% Brief explanation of experimental pipeline
The experiment consists of a two-pronged approach, where data acquisition and analysis are performed separately. Random initial states are generated from the ensemble of single-qubit 2-designs. In this case, the eigenstates of the Pauli operators are used. After $N$ random states are drawn, the outcome statistics of the chosen POVM are aggregated in a vector of relative frequencies $\hat{f}$. 
{These probe states are known and, in the presented experimental setting, efficiently prepared via single qubit pulses acting on the target qubits.}
With the data collected, post-processing starts by a least-squares estimator of each element of the POVM. This estimate, does not generally satisfy the POVM constraints, therefore a projection is needed to obtain the physical object we are interested in.

% Choices of POVMs and machine
We choose a set of representative POVMs  to showcase the capabilities of the algorithm. In this section we present results on a single-qubit Symmetric Informationally Complete (SIC-)POVM
that has been realized using one additional auxiliary qubit~\cite{jiang2020sic_circuit, you2024sic_basis}. That is, the setup involves two qubits in total. Appendix~\ref{app:further_results} contains additional performance studies for the two qubit Identity, Hadamard, Bell and Haar random measurements, as well as details on the implementation of the SIC-POVM.

% Results on SIC-POVMs (more on appendix)
The reconstruction of the experimental SIC-POVM is shown in Fig. \ref{fig:sicpovm}. This has been performed with a budget of $N=1.66\cdot10^5$ initial states, randomly sampled from a single qubit two-design. The reconstruction closely follows the target POVM. We also notice that there is a non-negligible difference between the experimental and target results. This difference is considerably larger than the precision of the reconstruction algorithm. This discrepancy is due to the experimental setup having inherent error in the gate application and readout that our protocol is precise enough to characterize.

\begin{figure}[t!]
    \centering
    \includegraphics[width=0.99\linewidth]{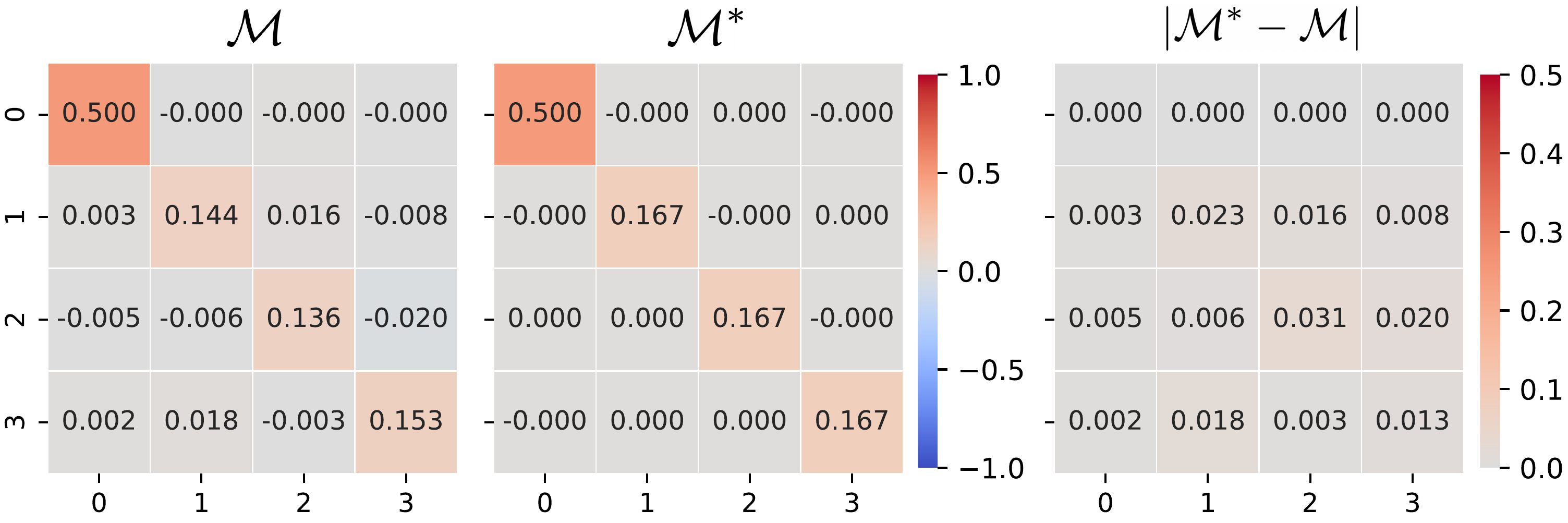
    }
    \caption{Comparison of the half-sided noisy measurement channel with the expected channel for a one-qubit SIC-POVM. The channels are represented in Pauli basis. The absolute difference between them is also shown on the right. The half-sided noisy channel only slightly deviates from the expected results, proportional to the depolarizing channel. This reconstruction of the channel actually implemented in the device can be used to improve results executed on the experimental device.}
    \label{fig:sicmeas}
\end{figure}

% Why SICs, possible applications
This estimated POVM can be used to construct, for example, the half-sided noisy measurement channel. This channel combines the noisy estimation of the POVM with the ideal one. It is computed by using the noisy and exact POVM reconstructions, $\{E_j^*\}_{j=1}^L$ and $\{E_j\}_{j=1}^L$ respectively, in
\begin{equation}
    \mathcal M^* = \sum_{j=1}^L |E_j)(E^*_j|.
\end{equation}
Here the round bra-ket follows the standard Dirac notation with respect to the Hilbert-Schmidt inner product $(A|B) = \mathrm{tr}[A^\dagger B]$ with $A, B$ linear operators \cite{brieger2023cgst}.
This half-sided measurement channel $\mathcal M^*$ for the SIC-POVM is shown in Fig. \ref{fig:sicmeas}, compared to the exact channel. We can see the effects of the noisy implementation on the recovered measurement channel. However, the difference between the expected and experimental results is small, and the proportionality to the effective depolarizing channel is still apparent. This channel is the key to error mitigation protocols for \textit{classical shadows} \cite{huang2020classicalshadow, maciejewski2020mitigation, chen2021robustshadow, koh2022classical, brieger2023cgst}. More precisely, the classical shadow protocol \cite{huang2020classicalshadow} relies on the inversion of the measurement channel of a spherical two design, which is proportional to the effective depolarizing channel. Of course, when implementing this POVM measurements on a noisy device, the resulting channel will not be the expected one. Therefore, the mismatch of the implemented channel and the one inverted will introduce some error on the expectation values of interest. The presented protocol, however, can estimate the noisy experimental POVM, which can then be used in place of the depolarizing channel \cite{chen2021robustshadow}. Several works have tackled this issue by estimating the noisy POVM implementation with different techniques \cite{maciejewski2020mitigation, van2022model, brieger2023cgst}, but our protocol offers both a practical and theoretical advancement on the previous methods, as it targets the noisy POVM implemented directly, making it a clear choice for the implementation of these robust classical shadow methods.

% Measurement tomography as calibration tool
Furthermore, as shown in the previous section, this protocol can be incorporated in calibration and characterization pipelines for quantum devices \cite{maciejewski2020mitigation}. Understanding the error sources of these early devices is crucial to perform useful computations. These can be used to design error mitigation schemes and develop targeted error correction frameworks. Readout error is often decoupled from the rest of the process, our protocol instead provides a complete picture for complex measurements, that includes logical operations and auxiliary space.

\section{Synopsis}

This work addresses the problem of quantum measurement tomography, which aims to reconstruct an unknown POVM based on the statistics obtained from a known set of input states. We propose and analyze a two-step reconstruction protocol that is both computationally fast and provably accurate. The protocol first constructs a linear least-squares estimator of the unknown POVM based on measurement statistics, and then projects this estimator onto the space of valid POVMs. This approach resembles projected least-squares quantum state and process tomography, but is adapted to the measurement problem \cite{guta2020fast, surawy2022projected}.

A key strength of our method is that the least-squares estimator can be computed analytically when the probe states are drawn from specific ensembles. We focus on two prominent choices: global 2-designs and tensor products of local 2-designs. For both ensembles, we provide rigorous, non-asymptotic performance guarantees for our estimator in terms of its average-case and worst-case reconstruction error.

In particular, for a $L$-outcome POVM acting on a $d$-dimensional system, our method requires $\mathcal{O}((d^3 + d^2 L)/\epsilon^2)$ samples to achieve error $\epsilon$ in operational (worst-case) distance with high probability when using global 2-designs, and $\mathcal{O}((12^n + 6^n L)/\epsilon^2)$ samples when using tensor products of single-qubit 2-designs. For average-case distance, the sample complexity is $\mathcal{O}(d^2 L / \epsilon^2)$ and $\mathcal{O}(5^n L / \epsilon^2)$, respectively.

In principle, a POVM can be viewed as a quantum-to-classical channel, meaning full reconstruction could be achieved via standard quantum process tomography using the diamond norm. However, estimating a quantum channel to accuracy $\epsilon$ in diamond distance requires $\mathcal{O}(d_{\mathrm{in}}^{3} L_{\mathrm{out}}^{3}/\epsilon^{2})$ samples, which is substantially larger than our results. The choice between our two proposed metrics presents a trade-off depending on the application. The operational distance $d_{\mathrm{op}}$ requires more samples but provides strict guarantees for all possible states, making it the necessary choice for rigorous error mitigation (as shown in Sec.~\ref{sec:application_to_error_mitigation}). On the other hand, the average-case distance ($d_{\mathrm{av}}$) requires fewer resources and is much better suited for randomized or near-term protocols (such as variational algorithms or classical shadows). In these typical experiments, worst-case bounds might overestimate the errors, making $d_{\mathrm{av}}$ the most efficient and natural choice.

To assess the optimality of our approach, we also derive lower bounds on the sample complexity of any non-collective, non-adaptive QMT protocol. Our proof combines techniques from quantum information theory and classical statistics, such as discretization arguments and Fano's inequality. We show that any protocol aiming to reconstruct an unknown POVM up to error $\epsilon$ must use at least $\Omega((d^3 + d^2L) / \epsilon^2)$ samples for the operational distance and $\Omega(d^2 L / \epsilon^2)$ samples for the average distance. These bounds match the scaling of our upper bounds for both distances, achieving strict optimality in both $d$ and $L$.

Beyond theoretical guarantees, our numerical results confirm that the protocol is practically viable for near-term devices. We show that the empirical reconstruction errors across various system dimensions ($n \leq 6$) and outcome counts ($L \leq 32$) consistently remain below our non-asymptotic upper bounds. Crucially, our approach exhibits stable runtime profiles, offering a computationally superior alternative to standard MLE tomography without a penalty in precision.

Finally, we complement our theoretical results with empirical studies on actual quantum hardware. Our experiments validate the practical performance of our protocol and demonstrate its robustness under real noise conditions. Moreover, we discuss the use of the protocol in the context of error mitigation and characterization of contemporary devices.

\begin{acknowledgments}
LZ was supported by the Government of Spain (Severo Ochoa CEX2019-000910-S, FUNQIP and European Union NextGenerationEU PRTR-C17.I1), Fundació Cellex, Fundació Mir-Puig, Generalitat de Catalunya (CERCA program), the EU Quantera project Veriqtas and the EU and Spanish AEI project QEC4QEA.
SRC would like to acknowledge the calibration team at Technology Innovation Institute (TII) for assistance with the experimental setup.
RK acknowledges financial support from the Austrian Science Fund (FWF) via the START award q-shadows (713361001) and the SFB BeyondC (10.55776/FG7) (research group 7).
\end{acknowledgments}

\bibliographystyle{apsrev4-2}
\bibliography{bib}

\appendix
\section{Estimator using $2$-designs}\label{app:2-designs}

\subsection{Least-squares estimator using $2$-designs}
The measurement step in the protocol defines a linear map $\mathcal{M}: M(\mathbb{C}^d) \rightarrow \mathbb{R}^{M}$ from POVM elements to probabilities. The map is defined as
\begin{align}
    [\mathcal{M} (X)]_{i} = \frac{1}{M} \langle \psi_i | X | \psi_i \rangle,
\end{align}
where $\{| \psi_i \rangle \}_{i=1}^{M}$ is a set of states that forms a $2$-design.  In the experiment, we uniformly sample a state from the set, prepare it, and then measure it using a POVM $E$ with elements $(E_1, E_2, \dots E_L)$. After performing several rounds of this procedure, the collected data consists of a vector of relative frequencies $\hat{f} \in \mathbb{R}^{M \times L}$, where
$\hat f_{ij}=\frac{N_{ij}}{N}$ and $N_{ij}$ is the number of shots in which the input state $|\psi_i\rangle$ was sampled and outcome $j$ was observed. The relation between the frequencies and the POVM elements is given by
\begin{align}
    \hat{f}_j = \mathcal{M}(E_j) + \varepsilon,
\end{align}
where $\hat{f}_{j} = (\hat{f}_{1j}, \hat{f}_{2j}, \dots, \hat{f}_{Mj})$ are all the frequencies associated to $E_j$ and $\varepsilon$ represents statistical noise due to the finite sample size.

The task of POVM estimation can be formulated as a least-squares problem. The goal is to find estimators $\hat{E}_j$ that minimize the difference between the observed data $\hat{f}_j$ and the model $\mathcal{M} (E_j)$. The least-squares solution is given by
\begin{align}
    \hat{E}_j = (\mathcal{M}^\dagger \mathcal{M})^{-1} (\mathcal{M}^\dagger (\hat{f_{j}})), \label{eq:lss_general}
\end{align}
where $\mathcal{M}^\dagger: \mathbb{R}^{M} \rightarrow M(\mathbb{C}^d)$ is the adjoint map of $\mathcal{M}$. 

To compute the solution, we leverage the 2-design property of the states $\{| \psi_i \rangle \}_{i=1}^{M}$. Specifically, it can be shown that the map $\mathcal{M}^\dagger \mathcal{M}$ satisfies \cite{guta2020fast}
\begin{align}
    \mathcal{M}^\dagger \mathcal{M} (X) &= \frac{1}{M^2} \sum_{i=1}^M \langle \psi_i | X | \psi_i \rangle |\psi_i \rangle \langle \psi_i | \nonumber \\
    & = \frac{X + \mathrm{tr}(X) \mathds{1}}{d(d+1)M}.
\end{align}
This result simplifies the inversion of $\mathcal{M}^\dagger \mathcal{M}$. The inverse channel is given by
\begin{align}
\left(\mathcal{M}^\dagger \mathcal{M} \right)^{-1} (X) = dM\left((d+1) X - \mathrm{tr}(X) \mathds{1}\right).
\end{align}
Substituting this into the least-squares solution, \eqref{eq:lss_general}, we find that the estimators $\hat{E}_j$ take the form 
\begin{align}
    \hat{E}_j &= \left(\mathcal{M}^\dagger \mathcal{M} \right)^{-1} (\mathcal{M}^\dagger (\hat{f}_j)) \nonumber \\
    &=    \left(\mathcal{M}^\dagger \mathcal{M} \right)^{-1} \left(\frac{1}{M}\sum_{i=1}^M  \hat{f}_{ij} | \psi_i \rangle \langle \psi_i | \right) \nonumber \\
    &= \sum_{i=1}^M  \hat{f}_{ij} \left( d(d+1) |\psi_i \rangle \langle \psi_i | - d \mathds{1} \right). 
\end{align}
Thus, the least-squares estimator provides a closed-form reconstruction of the POVM elements from the observed data.

\subsection{Scalar Bernstein parameters for global $2$-designs}
\label{app:global_2design_scalar_bernstein}

Here we compute the range and variance parameters used there for global complex projective $2$-design probe ensembles.

Fix a unit vector $|u\rangle \in \mathbb{C}^d$ and a sign vector
$a=(a_1,\ldots,a_L)\in\{\pm1\}^L$. Define
\begin{align}
    g_i(u)
    =
    \langle u|\nu_i|u\rangle,
    \qquad
    \nu_i=d(d+1)|\psi_i\rangle\langle\psi_i|-d\mathds{1},
\end{align}
and let $Z_{u,a}$ be the real-valued random variable such that
\begin{align}
    Z_{u,a}(i,j)
    =
    a_j g_i(u),
\end{align}
with probability  $p_{ij} = \frac{1}{M}\langle\psi_i|E_j|\psi_i\rangle$.

Using $\frac{1}{M}\sum_{i=1}^M \langle\psi_i|X|\psi_i\rangle\nu_i = X$, we have
\begin{align}
    \mathbb{E}[Z_{u,a}]
    &=
    \sum_{i=1}^M\sum_{j=1}^L
    \frac{1}{M}
    \langle\psi_i|E_j|\psi_i\rangle
    a_j
    \langle u|\nu_i|u\rangle
    \nonumber\\
    &=
    \sum_{j=1}^L a_j\langle u|E_j|u\rangle .
\end{align}
Since $E$ is a POVM,
\begin{align}
    |\mathbb{E}Z_{u,a}|
    \leq
    \sum_{j=1}^L \langle u|E_j|u\rangle
    =
    1.
    \label{eq:global_expectation_range_app}
\end{align}

We now compute the second moment of $g_i(u)$. Set $x_i=|\langle u|\psi_i\rangle|^2$. Then
\begin{align}
    g_i(u)=d(d+1)x_i-d.
\end{align}
Because $\{|\psi_i\rangle\}_{i=1}^M$ is a complex projective $2$-design,
\begin{align}
    \frac{1}{M}\sum_{i=1}^M x_i
    =
    \frac1d,
    \qquad
    \frac{1}{M}\sum_{i=1}^M x_i^2
    =
    \frac{2}{d(d+1)} .
\end{align}
Therefore
\begin{align}
    \frac{1}{M}\sum_{i=1}^M g_i(u)^2
    &=
    d^2(d+1)^2
    \frac{1}{M}\sum_{i=1}^M x_i^2
      \nonumber \\
     & \quad \quad -
    2d^2(d+1)
    \frac{1}{M}\sum_{i=1}^M x_i
    +
    d^2
    \nonumber\\
    &=
    d^2 .
    \label{eq:global_g_second_moment_app}
\end{align}
Moreover, since $0\leq x_i\leq 1$, $|g_i(u)|\leq d^2$. Then, the centered random variable satisfies
\begin{align}
    |Z_{u,a}-\mathbb{E}Z_{u,a}|
    \leq
    d^2+1 .
    \label{eq:global_centered_range_scalar_app}
\end{align}

The variance of $Z_{u,a}$ satisfies
\begin{align}
    \operatorname{Var}(Z_{u,a})
    &\leq
    \mathbb{E}Z_{u,a}^2
    \nonumber\\
    &=
    \sum_{i=1}^M\sum_{j=1}^L
    \frac{1}{M}
    \langle\psi_i|E_j|\psi_i\rangle
    a_j^2 g_i(u)^2
    \nonumber\\
    &=
    \frac{1}{M}\sum_{i=1}^M g_i(u)^2
    \nonumber\\
    &=
    d^2 ,
    \label{eq:global_variance_scalar_app}
\end{align}
where we used $a_j^2=1$ and $\sum_jE_j=\mathds{1}$.

Thus, for global $2$-design probe states, the scalar Bernstein parameters used in the operational-distance proof are
\begin{align}
    v^2=d^2,
    \qquad
    K=d^2+1 .
\end{align}

\section{Estimator using single-qubit $2$-designs}\label{app:single_qubit-2-designs}

\subsection{Least-squares estimator using single-qubit $2$-designs}
 We now assume we have an $n$-qubit system and that the set of states $\{|\psi_i \rangle\}_{i=1}^M$ we measure is a tensor product of $n$ single-qubit sets $\{|\psi_{i_k} \rangle \}_{i_k=1}^{m}$ that form a $2$-design. Then, we have $M=m^n$ elements of the form
\begin{align}
    |\psi_i \rangle = |\psi_{i_1} \psi_{i_2} \dots \psi_{i_n}\rangle,
\end{align}
Since for each qubit we have a local $2$-design, we can define reduced channels $\mathcal{M}_k: M(\mathbb{C}^2) \rightarrow \mathbb{R}^{m}$ for each qubit $k$. Then, the map $\mathcal{M}_{k}^\dagger \mathcal{M}_{k}$  can be calculated as
\begin{align}
    \mathcal{M}_k^\dagger \mathcal{M}_k (X)  = \frac{X + \mathrm{tr}(X) \mathds{1}}{6m},
\end{align}
with inverse
\begin{align}
    \left(\mathcal{M}_k^\dagger \mathcal{M}_k \right)^{-1} (X) = 2m\left(3 X - \mathrm{tr}(X) \mathds{1}\right).
\end{align}
Following Ref. \cite{guta2020fast}, the total channel is just a tensor product of the local channels, 
\begin{align}
    \mathcal{M}^\dagger \mathcal{M} (X) =  \bigotimes_{k=1}^n  \frac{X + \mathrm{tr}(X) \mathds{1}}{6m}
\end{align}
with inverse
\begin{align}
    \left(\mathcal{M}^\dagger \mathcal{M}  \right)^{-1}(X) &=  \bigotimes_{k=1}^n  2m\left(3 X - \mathrm{tr}(X) \mathds{1}\right) \nonumber \\
    &= M \bigotimes_{k=1}^n \left(6 X - 2\mathrm{tr}(X) \mathds{1}\right).
\end{align}
Applying the formula for the least squares estimator, we obtain
\begin{align}
    \hat{E}_j &= (\mathcal{M}^\dagger \mathcal{M})^{-1} \left(\mathcal{M}^\dagger \left(\hat{f_{j}} \right)\right) \nonumber \\
    &=    \left(\mathcal{M}^\dagger \mathcal{M} \right)^{-1} \left(\frac{1}{M}\sum_{i=1}^M  \hat{f}_{ij} | \psi_i \rangle \langle \psi_i | \right) \nonumber \\
    &=  \sum_{i=1}^M  \hat{f}_{ij} \bigotimes_{k=1}^n\left(6 | \psi_{i_k} \rangle \langle \psi_{i_k} | -  2 \mathds{1}\right). 
\end{align}

\subsection{Scalar Bernstein parameters for tensor products of single-qubit $2$-designs}
\label{app:local_2design_scalar_bernstein}

We now compute the corresponding scalar Bernstein parameters when the probe ensemble is a tensor product of single-qubit $2$-designs. 

As before, fix a unit vector $|u\rangle$ and a sign vector
$a\in\{\pm1\}^L$, Define
\begin{align}
    g_i(u)=\langle u|\nu_i|u\rangle,
    \qquad
        \nu_i
    =
    \bigotimes_{k=1}^n
    \left(
        6|\psi_{i_k}\rangle\langle\psi_{i_k}|
        -
        2\mathds{1}
    \right).
\end{align}
and let $Z_{u,a}$ be the real-valued random variable such that 
\begin{align}
    Z_{u,a}(i,j)=a_jg_i(u)
\end{align}
with probability $p_{ij} = \frac{1}{M}\langle\psi_i|E_j|\psi_i\rangle$.

The expectation is
\begin{align}
    \mathbb{E}[Z_{u,a}]
    =
    \sum_{j=1}^L a_j\langle u|E_j|u\rangle,
\end{align}
and therefore $|\mathbb{E}Z_{u,a}| \leq 1$.

We first bound the range of $g_i(u)$. For each qubit,
\begin{align}
    \left\|
        6|\psi_{i_k}\rangle\langle\psi_{i_k}|
        -
        2\mathds{1}
    \right\|
    =
    4.
\end{align}
Hence
\begin{align}
    |g_i(u)|
    \leq
    \|\nu_i\|
    \leq
    4^n .
    \label{eq:local_g_range_app}
\end{align}
This gives the centered range bound
\begin{align}
    |Z_{u,a}-\mathbb{E}Z_{u,a}|
    \leq
    4^n+1 .
    \label{eq:local_centered_range_scalar_app}
\end{align}

It remains to compute a useful second-moment bound for $g_i(u)$. Let
$\mathcal{P}_n$ denote the $n$-qubit Pauli strings. Then, we can expand
\begin{align}
    |u\rangle\langle u|
    =
    \frac1d\sum_{P\in\mathcal{P}_n} r_P P,
    \qquad
    r_P=\operatorname{tr}\!\left(P|u\rangle\langle u|\right).
\end{align}
Since $|u\rangle\langle u|$ is pure $\sum_{P\in\mathcal{P}_n} r_P^2 = d$
Expand
\begin{align}
    \nu_i
    =
    \sum_{P\in\mathcal{P}_n} c_P(i)P .
\end{align}
Then
\begin{align}
    g_i(u)
    =
    \operatorname{tr}\!\left(|u\rangle\langle u|\,\nu_i\right)
    =
    \sum_{P\in\mathcal{P}_n} r_P c_P(i).
    \label{eq:g_pauli_coeff_app}
\end{align}

For a single qubit, if the local probe state is a Pauli eigenstate
$|\psi_{\alpha,s}\rangle$ with $\alpha\in\{x,y,z\}$ and $s\in\{\pm1\}$, then
\begin{align}
    6|\psi_{\alpha,s}\rangle\langle\psi_{\alpha,s}|-2\mathds{1}
    =
    \mathds{1}+3s\sigma_\alpha .
\end{align}
Averaging uniformly over the six states gives
\begin{align}
    \mathbb{E}[c_{\mathds{1}}^2]=1,
    \quad
    \mathbb{E}[c_{\sigma_\alpha}c_{\sigma_\beta}]
    =
    3\delta_{\alpha,\beta},
    \quad
    \mathbb{E}[c_{\mathds{1}}c_{\sigma_\alpha}]=0 .
\end{align}
Tensorizing over $n$ qubits yields
\begin{align}
    \frac{1}{M}\sum_{i=1}^M c_P(i)c_Q(i)
    =
    \delta_{P,Q}3^{w(P)},
    \label{eq:local_coeff_covariance_app}
\end{align}
where $w(P)$ is the Pauli weight of $P$.

Using Eqs.~\eqref{eq:g_pauli_coeff_app} and
\eqref{eq:local_coeff_covariance_app}, we obtain
\begin{align}
    \frac{1}{M}\sum_{i=1}^M g_i(u)^2
    &=
    \frac{1}{M}\sum_{i=1}^M
    \left(
        \sum_{P\in\mathcal{P}_n} r_Pc_P(i)
    \right)^2
    \nonumber\\
    &=
    \sum_{P,Q\in\mathcal{P}_n}
    r_Pr_Q
    \frac{1}{M}\sum_{i=1}^M c_P(i)c_Q(i)
    \nonumber\\
    &=
    \sum_{P\in\mathcal{P}_n}
    3^{w(P)}r_P^2
    \nonumber\\
    &\leq
    3^n\sum_{P\in\mathcal{P}_n}r_P^2
    \nonumber\\
    &=
    3^n d
    =
    6^n .
    \label{eq:local_g_second_moment_app}
\end{align}

The variance of $Z_{u,a}$ is then bounded as
\begin{align}
    \operatorname{Var}(Z_{u,a})
    &\leq
    \mathbb{E}Z_{u,a}^2
    \nonumber\\
    &=
    \sum_{i=1}^M\sum_{j=1}^L
    \frac{1}{M}
    \langle\psi_i|E_j|\psi_i\rangle
    a_j^2 g_i(u)^2
    \nonumber\\
    &=
    \frac{1}{M}\sum_{i=1}^M g_i(u)^2
    \nonumber\\
    &\leq
    6^n .
    \label{eq:local_variance_scalar_app}
\end{align}
Therefore, for tensor products of single-qubit $2$-designs, the scalar Bernstein parameters used in the operational-distance proof are
\begin{align}
    v^2=6^n,
    \qquad
    K=4^n+1 .
\end{align}

\section{Proof of Thm.~\ref{thm:upper_bound_d_av}}\label{app:proof_thm_4}

In the proof, we will need the following theorem:

\begin{lemma}[Vector Bernstein Inequality \cite{pinelis1994optimum, martinez2024empirical}]\label{lem:vector_bernstein}
    Let $\vec{Z}_1, \dots, \vec{Z}_N$ be independent, zero-mean random vectors in a Hilbert space $\mathcal{H}$. Suppose that $\Vert \vec{Z}_k \Vert_{\mathcal{H}} \leq K$ and $\frac{1}{N} \sum_{k=1}^N \mathbb{E}\Vert \vec{Z}_k \Vert_{\mathcal{H}}^2 \leq \sigma^2$. Then for any $t > 0$,
    \begin{align}
        \mathrm{Pr}\left( \left\Vert \frac{1}{N} \sum_{k=1}^N \vec{Z}_k \right\Vert_{\mathcal{H}} \geq t \right) \leq 2 \exp \left( - \frac{N t^2 / 2}{\sigma^2 + K t / 3} \right).
    \end{align}
\end{lemma}

Then, we can prove the following:

\begin{reptheorem}{thm:upper_bound_d_av}
    Assume an $L$-outcome POVM $E$ acting on a Hilbert space of dimension $d$ and an IC set of states $\{|\psi_i \rangle \}_{i=1}^{M}$. A reconstruction error $d_{\mathrm{av}}(E, E^*) \leq \epsilon$ with probability $1-\delta$ can be achieved using the protocol described above with sample size
    \begin{align}
        N &\geq \frac{ 8L\left(d^2 + d + \frac{d\epsilon}{3\sqrt{L}} \right)}{ \epsilon^2}  \ln \left(\frac{4}{\delta}\right)
    \end{align}
    when $\{|\psi_i\rangle \}_{i=1}^{M}$ forms a $2$-design, or
    \begin{align}
        N \geq \frac{8L\left(5^n + \frac{5^{n/2}\epsilon}{3\sqrt{L}}\right)}{\epsilon^2}  \ln \left(\frac{4}{\delta}\right)
    \end{align}
    if the POVM acts on an $n$-qubit system ($d=2^n$) and $\{|\psi_i\rangle \}_{i=1}^{M}$ is a tensor product of single-qubit $2$-designs.
\end{reptheorem}

\begin{proof}
Let $W_j=\hat E_j-E_j$. We first control the error of the unconstrained
least-squares estimator. Using the definition of $d_{\mathrm{av}}$, extended
to arbitrary Hermitian $L$-tuples, we have
\begin{align}\label{eq:dav_split}
    d_{\mathrm{av}}(E,\hat E)
    &=
    \frac{1}{2d}
    \sum_{j=1}^L
    \sqrt{
        \|W_j\|_F^2+
        \left(\operatorname{tr}W_j\right)^2
    }
    \nonumber\\
    &\le
    \frac{1}{2d}\sum_{j=1}^L\|W_j\|_F
    +
    \frac{1}{2d}\sum_{j=1}^L|\operatorname{tr}W_j| .
\end{align}
    Let $T_1 = \frac{1}{2d} \sum_{j=1}^L \Vert W_j \Vert_F$ and $T_2 = \frac{1}{2d} \sum_{j=1}^L \left| \mathrm{tr}[W_j] \right|$. We will show that, with probability at least $1-\delta$, both terms are at most $\epsilon/4$. This implies $d_{\mathrm{av}}(E,\hat E)\le \epsilon/2$.

    By Cauchy-Schwarz inequality, we have
    \begin{align}\label{eq:T1_CS}
        T_1 \leq \frac{\sqrt{L}}{2d} \sqrt{\sum_{j=1}^L \Vert W_j \Vert_F^2} = \frac{\sqrt{L}}{2d} \Vert \vec{W} \Vert_{\mathcal{H}},
    \end{align}
    where $\vec{W} = (\hat{E}_1 - E_1, \dots, \hat{E}_L - E_L)$ is treated as a single vector in the Hilbert space $\mathcal{H} = (\mathbb{C}^{d \times d})^{\oplus L}$ with the  norm $\Vert \vec{V} \Vert_{\mathcal{H}}^2 = \sum_{j=1}^L \Vert V_j \Vert_F^2$.

The least-squares estimator can be written as $\hat{\vec E} = \frac1N\sum_{k=1}^N \vec X_k$.
For each shot $k$, we sample an input state $|\psi_{i_k}\rangle$ and observe
an outcome $j_k$. The random vector $\vec X_k\in\mathcal H$ has a single
nonzero block, namely $(X_k)_j = \nu_{i_k}\,\mathds 1\{j_k=j\}$.
By construction, $\mathbb E[\vec X_k]=\vec E$, where $\vec E=(E_1,\ldots,E_L)$. Hence
\begin{align}
    \vec W
    =
    \frac1N\sum_{k=1}^N \vec Z_k,
    \qquad
    \vec Z_k:=\vec X_k-\vec E,
\end{align}
and the vectors $\vec Z_k$ are independent and mean-zero.

For a global $2$-design, $\|\vec X_k\|_{\mathcal H}^2 = \|\nu_{i_k}\|_F^2 = d^4+d^3-d^2$.
Moreover,
\begin{align}
    \|\vec E\|_{\mathcal H}^2
    \le
    \sum_{j=1}^L(\operatorname{tr}E_j)^2
    \le
    \left(\sum_{j=1}^L\operatorname{tr}E_j\right)^2
    =
    d^2 .
\end{align}
Consequently,
\begin{align}
    \|\vec Z_k\|_{\mathcal H}
    \le
    \|\vec X_k\|_{\mathcal H}+\|\vec E\|_{\mathcal H}
    \le
    2d^2 ,
\end{align}
for $d\ge2$. Thus, we may take $K=2d^2$ in vector Bernstein.

The variance parameter satisfies
\begin{align}
    \mathbb E\|\vec Z_k\|_{\mathcal H}^2
    &\le
    \mathbb E\|\vec X_k\|_{\mathcal H}^2
    \nonumber\\
    &=
    \sum_{i=1}^M\sum_{j=1}^L
    \frac{1}{M}
    \langle\psi_i|E_j|\psi_i\rangle
    \|\nu_i\|_F^2
    \nonumber\\
    &\le
    d^4+d^3 .
\end{align}
Here we used $\sum_jE_j=\mathds{1}$. Therefore  $\sigma^2=d^4+d^3$

To guarantee $T_1\le\epsilon/4$, it is enough by
Eq.~\eqref{eq:T1_CS} to require $\|\vec W\|_{\mathcal H} \le \frac{d\epsilon}{2\sqrt L}$. Applying vector Bernstein gives
\begin{align}
    \Pr\left(T_1>\frac{\epsilon}{4}\right) \le 2\exp\left(\frac{ - N d^2\epsilon^2/(8L)}{d^4+d^3+d^3\epsilon/(3\sqrt L)}
    \right).
\end{align}
Thus $T_1\le\epsilon/4$ with probability at least $1-\delta/2$ provided
\begin{align}
    N
    \ge
    \frac{
        8L\left(d^2+d+\frac{d\epsilon}{3\sqrt L}\right)
    }{\epsilon^2}
    \ln\left(\frac4\delta\right).
\end{align}

We now bound $T_2$. For both global and local $2$-design frames, $\operatorname{tr}\nu_i=d$. Therefore
\begin{align}
    \operatorname{tr}\hat E_j
    =
    \frac dN\sum_{k=1}^N \mathds 1\{j_k=j\}
    =
    d\hat p_j,
\end{align}
where $\hat p_j$ is the empirical frequency of outcome $j$. Similarly, $p_j =\frac{\operatorname{tr}E_j}{d}$
is the corresponding marginal outcome probability, obtained by averaging
the input state over the $2$-design. Hence
\begin{align}
    T_2
    =
    \frac12\sum_{j=1}^L|\hat p_j-p_j|.
\end{align}
Using Cauchy--Schwarz, $T_2 \le \frac{\sqrt L}{2}\|\hat{\vec p}-\vec p\|_2$. Let $\vec v_k=\vec e_{j_k}-\vec p$. Then $\mathbb E[\vec v_k]=0$,
\begin{align}
    \|\vec v_k\|_2\le2,
    \qquad
    \mathbb E\|\vec v_k\|_2^2\le1 .
\end{align}
Applying vector Bernstein in $\mathbb R^L$ shows that $T_2\le\epsilon/4$ with probability at least $1-\delta/2$ provided
\begin{align}
    N
    \ge
    \frac{
        8L\left(1+\frac{\epsilon}{3\sqrt L}\right)
    }{\epsilon^2}
    \ln\left(\frac4\delta\right).
\end{align}
Taking the maximum of the constraints for $T_1$ and $T_2$ yields the bound $\epsilon/2$ for the global 2-design.

An analogous derivation holds for local 2-designs. The local frame elements satisfy $\Vert \nu_i \Vert_F^2 = 5^n d^2 = \sigma^2$ and $ \|\vec Z_k\|_{\mathcal H}  \leq 20^{n/2} + 2^n \leq  2\; 5^{n/2} d = K$. Substituting these values into the Vector Bernstein bound gives the required sample complexity for tensor products of single-qubit designs.

    Finally, due to the operational properties of the projection step, we have $d_{\mathrm{av}}(E, E^*) \leq 2 d_{\mathrm{av}}(E, \hat{E})$. Therefore, ensuring $d_{\mathrm{av}}(E, \hat{E}) \leq \epsilon/2$ guarantees $d_{\mathrm{av}}(E, E^*) \leq \epsilon$, concluding the proof.
\end{proof}

\section{Minimum sample complexity for $d_{\mathrm{av}}$}\label{app:minimal_sample_complexity_dav}

  Here, we provide a lower bound for the sample complexity of non-adaptive measurement tomography considering $d_{\mathrm{av}}$. The proof is based on the discretization of the problem, reducing quantum measurement tomography to the problem of discrimination of well-separated POVMs.  We first construct a set of $R \in \exp(\Omega(d^2 L))$ POVMs on dimension $d$ that are $\epsilon/4$ apart in $\sqrt{d} \, d_{\mathrm{av}}$ from each other. We then encode a random message using this set and decode it using measurement tomography with sufficient precision. From Fano's inequality, this gives us a lower bound $\Omega(d^2 L)$ for the mutual information between the encoder and decoder. Additionally, we obtain an upper bound $\mathcal{O}(N \epsilon^2 / d)$ for the mutual information between the parties after $N$ uses of the POVM. Using these two results, we derive a bound $N \geq \Omega(d^2 L /\epsilon^2)$ for the sample complexity of any non-adaptive tomographic procedure using $d_\mathrm{av}$.

\subsection{Construction of an $\epsilon$-packing}

Let $\{U_j\}_{j=1}^{L/2}$ be a set of Haar-random unitaries, $P$ a rank-$d/2$ projector and $E_{\{U_j\}}$ a POVM with $L$ elements given by
\begin{align}\label{eq:POVM_dav}
    E^{j} &= \frac{(1-\epsilon)}{L}\mathds{1} + \frac{2\epsilon}{L} U_j P U_j^\dagger \nonumber\\
    E^{j+L/2} & = \frac{(1 + \epsilon)}{L}\mathds{1} - \frac{2\epsilon}{L} U_j P U_j^\dagger
\end{align}
 We want to construct a large set of POVMs of this form such that $\sqrt{d} \,d_{\mathrm{av}} (E_{\{U_j\}}, E_{\{V_j\}}) \geq  \epsilon/4$. To do this,  we will use the following \cite{Meckes2013spectral}:

\begin{theorem}\label{thm:meckes}
Let $k$ and $d$ be positive integers and $M = U_1(d) \times \dots \times U_k(d)$ a space equipped with the $L_2$-sum of Hilbert–Schmidt metrics. Suppose $F:M \rightarrow \mathbb{R}$ is $\kappa$-Lipschitz, and
that $\{U_j \in U_j(d)\}$ are independent, Haar distributed unitary matrices. Then,
for each $t > 0$,
\begin{align}
   \mathrm{Pr} & \left(  F(U_1, \dots, U_k) \geq \mathbb{E} \left[F(U_1, \dots, U_k)\right] + t \right) \nonumber \\
    &\leq \exp \left(- \frac{dt^2}{12 \kappa^2} \right). 
\end{align}
\end{theorem}

This theorem gives us a concentration inequality for $d_{\mathrm{av}}$. Then, if we choose our set of POVMs randomly, it will be $\epsilon$-separated with high probability.

\begin{lemma}\label{thm:packing_dav}
    There exists a set of $R \in \exp(\Omega (d^2 L))$ POVMs that are $\epsilon/4$ apart in the distance $\sqrt{d} d_{\mathrm{av}}$ from each other.
\end{lemma}

\begin{proof}
We want to use Theorem \ref{thm:meckes} for
\begin{align}
    F &= \frac{1}{2d}\sum_{i=1}^{L} \Vert E_{\{U_j\}}^{i} -E_{\{V_j\}}^{i}\Vert_F \nonumber\\
    &= \frac{\epsilon }{d L}\sum_{j=1}^{L} ||U_j P U_j^\dagger - V_j P V_j^\dagger||_F,
\end{align}
since this is a lower bound for $d_{\mathrm{av}}(E_{\{U_j\}}, E_{\{V_j\}})$. Defining $f_j = ||U_j P U_j^\dagger - V_j P V_j^\dagger||_F$, we need to calculate the Lipschitz constant $\kappa$ (from the relation $|F- F'| \leq \kappa \Vert (U_i, V_i) - (U_i', V_i')\Vert$) and $\mathbb{E}\left[\frac{\epsilon}{dL}\sum_{j=1}^{L} f_j\right]$. We start with the former. Using the reverse triangle inequality, we have 
\begin{align}
    &\big|f_j  - f_j' \big|\\
    & \; \leq  \Vert  U_j P U_j^\dagger - V_j P V_j^\dagger - U_j' P U_j'^\dagger + V_j' P V_j'^\dagger \Vert_F  \nonumber \\
   & \; \leq \Vert U_j P U_j^\dagger - U_j' P U_j'^\dagger \Vert_F + \Vert V_j P V_j^\dagger - V_j' P V_j'^\dagger \Vert_F. \nonumber
\end{align}
Then, since for unitaries $W$ and $W'$ we have the relation $\Vert W_j P W_j^\dagger - W_j' P W_j'^\dagger \Vert_F \leq 2 \left\Vert  W_j  - W_j'  \right\Vert_F$ and the unitaries satisfy $W_j = W_{j+L/2}$ for $j=1, \dots L/2$, we obtain 
\begin{align}
    &\left| \sum_{j=1}^{L}f_j  - f_j' \right|  \nonumber\\
    & \quad \leq  4 \sqrt{L} \left(\sum_{j=1}^{L/2}\left\Vert  U_j  - U_j'  \right\Vert_F^2 + \left\Vert  V_j  - V_j' \right\Vert_F^2  \right)^{\frac{1}{2}}.
\end{align}
This gives us an upper bound for the Lipschitz constant, $\kappa \leq \frac{4\epsilon}{d\sqrt{L}}$.

Now, we find a lower bound for the average $\mathbb{E}\left[\frac{1}{L}\sum_{j=1}^{L} f_j\right]$ using \cite{maciejewski2023exploring}
\begin{align}
    \mathop{\mathbb{E}}_{U, V \sim Haar} f_j \geq & \frac{\left(\mathbb{E}_{U, V} f_j^2 \right)^{3/2}}{\left(\mathbb{E}_{U, V} f_j^4 \right)^{1/2}}. 
\end{align}
We have
\begin{align}
    \mathbb{E}  f_j^2 = &  \mathbb{E} \mathrm{tr} \left(\left(  U_j P U_j^\dagger  - V_j PV_j^{\dagger} \right)^2 \right)\nonumber \\
    = & 2 \mathbb{E}\mathrm{tr}(P^2) - 2\mathbb{E} \mathrm{tr}  \left(U_j P U_j^\dagger V_j PV_j^{\dagger} \right) \nonumber \\
    = & 2\left(\mathrm{tr}(P^2) - \frac{\mathrm{tr}(P)^2}{d}\right).
\end{align}
Also, 
\begin{align}
   \mathbb{E}  f_j^4 = &  4 \mathbb{E} \, \left(  \mathrm{tr} \left(\left(  U_j P U_j^\dagger  - V_j PV_j^{\dagger} \right)^2 \right)\right)^2 \nonumber \\
    =& 4 \mathrm{tr} \left( P^2\right)^2 - 8 \frac{ \mathrm{tr}\left(P^2\right) \mathrm{tr}(P)^2}{d} \nonumber \\
    &+ 4\frac{ \mathrm{tr}\left(P\right)^4 + \mathrm{tr}\left(P^2\right)^2}{d^2-1} - 8 \frac{ \mathrm{tr}\left(P^2\right) \mathrm{tr}\left(P\right)^2}{d(d^2-1)}.
\end{align}
Since the projector $P$ has rank $d/2$, we get
\begin{align}
    \mathbb{E}  f_j^2 = \frac{d}{2}, \qquad \mathbb{E}  f_j^4 = \frac{d^4}{4(d^2-1)}.
\end{align}
Then, $\mathbb{E} \left(\frac{\epsilon}{d L}\sum_{j=1}^{L} ||U_j P U_j^\dagger - V_j P V_j^\dagger||_F \right) \geq \frac{\epsilon}{2\sqrt{d}}$.
Using Theorem \ref{thm:meckes}, we have
\begin{align}
    &\mathrm{Pr} \left( \frac{\epsilon}{d L}\sum_{j=1}^{L} ||U_j P U_j^\dagger - V_j P V_j^\dagger||_F \leq \frac{\epsilon}{4\sqrt{d}} \right) \nonumber \\
    &\leq  \exp \left(- \frac{d^2 L}{12\times 64 \times 16} \right). 
\end{align}

Then, using probabilistic arguments, we can build a set of  $R \in \exp(\Omega(d^2 L))$ POVMs 
\begin{align}
    E^{j} &= \frac{(1-\epsilon)}{L}\mathds{1} + \frac{2\epsilon}{L} U_j P U_j^\dagger\nonumber \\
    E^{j+L/2} &= \frac{(1+\epsilon)}{L} \mathds{1}- \frac{2\epsilon}{L} U_j P U_j^\dagger
\end{align}
that satisfy
\begin{align}
     \sqrt{d} \; d_{\mathrm{av}}(E, F)  \geq \frac{\epsilon}{4}. 
\end{align}

\end{proof}

\subsection{Upper bound on mutual information for the average distance}

 \begin{lemma}\label{thm:upper_bound_by_haar}
     Let $X \sim Unif[R]$,  $\{U_j\}_{j=1}^{L/2}$ be a set of independent Haar random unitaries and $\{\rho_i\}_{i=1}^N$ a set of $N$ quantum states. Then, there exists a set of $R \in \exp(\Omega (d^2 L))$ POVMs of the form of Eq. \eqref{eq:POVM_dav} which forms an $\epsilon/4$ packing for $\sqrt{d} \, d_{\mathrm{av}}$ and satisfies 
     \begin{align}
     I(X: Y ) \leq I(\{U_j\}: Z), 
     \end{align}
     where $Y = (Y_1, \dots Y_N)$ is the outcome of measuring $\{\rho_i\}_{i=1}^N$ with a random POVM $E_{\{U_j^{X}\}}$ and $Z = (Z_1, \dots Z_N)$ is the outcome of measuring $\{\rho_i\}_{i=1}^N$ with a random POVM $E_{\{U_j\}}$. 
 \end{lemma}

\begin{proof}
    We closely follow Proposition 4.3 in Ref.~\cite{lowe2022lower}. Let $X$ label a POVM of the form of Eq. \eqref{eq:POVM_dav}. The average distance for two POVMs of this form of is invariant under the replacement $U_j^{X} \rightarrow W_j U_j^{X}$ for arbitrary unitary operators $\{W_j\}_{j=1}^{L/2}$, since for two unitaries $U$ and $V$, we have $\Vert U P U^{\dagger} - V P V^{\dagger} \Vert_F = \Vert W U P U^\dagger W^{\dagger } -W V P V^\dagger W^{\dagger } \Vert_F$.  Then, this replacement maps an arbitrary initial $\epsilon/4$-packing to another one. 
    
    Define $Y_{W} = (Y_1, \dots Y_N)$, with $Y_i \in \{1, \dots, L\}$, as the outcome random variable obtained by measuring the $N$ copies $\{\rho_i\}_{i=1}^N$ with $E_{\{ W_j U^{X}_j \}}$. We claim that
    \begin{align}
        \mathop{\mathbb{E}}_{\{W_j \sim \mathrm{Haar}\}} I(X: Y_W) \leq I(\{U_j\} :Z)
    \end{align}
    for a set of independent, Haar distributed unitaries $\{W_j \}_{j=1}^{L/2}$. Let $p_{Y|W, X}$ be the distribution of $Y_W$ given $\{U_j^{X}\}$, with  $p_{Y|W, X}(Y=y) = \prod_{i=1}^N \mathrm{tr} \left(E_{W_{y_i} U_{y_i}^{X}} \rho_i  \right)$. We have
    \begin{align}
        &\mathop{\mathbb{E}}_{\{W_j \sim \mathrm{Haar}\}} I(X: Y_W)  \nonumber \\
        & \; \; = \mathop{\mathbb{E}}_{W_j} H \left(\mathop{\mathbb{E}}_{X \sim [R]} p_{Y|W, X} \right) 
        - \mathop{\mathbb{E}}_{W_j} \mathop{\mathbb{E}}_{X \sim [R]}   H \left( p_{Y|W, X} \right) \nonumber \\
        & \; \; \leq   H \left(\mathop{\mathbb{E}}_{X \sim [R]}    \mathop{\mathbb{E}}_{W_j} p_{Y|W, X} \right) -  \mathop{\mathbb{E}}_{X \sim [R]}    \mathop{\mathbb{E}}_{W_j} H \left( p_{Y|W, X} \right). \label{eq:mutual_inf_haar}
    \end{align}
The first line is just the definition of the mutual information, and the second one follows from the concavity of entropy. 
By the invariance of the Haar measure, we have that $\mathop{\mathbb{E}}_{\{W_j \sim \mathrm{Haar}\}} p_{Y|W, X} = \mathop{\mathbb{E}}_{\{W_j \sim \mathrm{Haar}\}} p_{Y|W}$ and $\mathop{\mathbb{E}}_{\{W_j \sim \mathrm{Haar}\}} H \left( p_{Y|W, X} \right) = \mathop{\mathbb{E}}_{\{W_j \sim \mathrm{Haar}\}} H \left( p_{Y|W} \right)$. The probability distribution of the outcomes of the random POVM $E_{\{U_j\}}$ is defined as  $p_Z = \mathop{\mathbb{E}}_{\{U_j \sim \mathrm{Haar}\}} p_{Y|U} = \mathop{\mathbb{E}}_{\{W_j \sim \mathrm{Haar}\}} p_{Y|W}$. Using Eq.~\eqref{eq:mutual_inf_haar}, this means that 
    \begin{align}
        \mathop{\mathbb{E}}_{\{W_j \sim \mathrm{Haar}\}} I(X: Y_W) \leq I(\{U_j\}:Z). 
    \end{align}
Since the expectation of $I(X: Y_W)$ over all sets of unitary operators $\{W_j\}_{j=1}^{L/2}$ is at most $I(\{U_j\}:Z)$, there exists at least one set of unitary operators $\{ V_j \}_{j=1}^{L/2}$ for which the inequality $I(X: Y_V) \leq  I(\{U_j\}: Z)$ holds. Then we consider the set of POVMs $\{ E^{\{V_i U_i^X\}}\}$. 
\end{proof}

This means that in the calculation of the mutual information we can replace the average over the unitaries that define the $\epsilon/4$-packing by an average over Haar random unitaries to obtain an upper bound. This is given by \cite{lowe2022lower}
\begin{lemma}\label{thm:upper_bound_mutual_information_dav}
    Let $X \sim Unif([R])$ and $Y = (Y_1, Y_2, \dots Y_N)$ be the outcome of the measurement of a random POVM $E_{\{U_X\}}$ of the form of Eq.~\eqref{eq:POVM_dav} over $N$ quantum states. Then, 
\begin{align}
    I(X: Y) \leq \frac{2 N}{\ln(2)} \frac{\epsilon^2}{(d+1)}
\end{align}
\end{lemma}
\begin{proof}
We have
\begin{align}
    I(X: Y) \leq & I(\{ U_j \} : Z) \\
    = & \sum_{i=1}^N I(\{U_j\}:Z_i|Z_{i-1}, \dots Z_{1})\\
    \leq & \sum_{i=1}^N I(\{U_j\}:Z_i) \\
    \leq & \frac{N}{\ln (2)}  \mathop{\mathbb{E}}_{\{ U_j\} } \left(\sum_{z_i} \frac{p_{Z_i|U_{z_i}} (z_i)^2}{p_{Z_i}(z_i)} -1\right), 
\end{align}
where the first inequality follows from Lemma \ref{thm:upper_bound_by_haar}, the second one from the chain rule for mutual information, the third one from the independence of $Z_1, \dots Z_N$ given $\{ U_j \}$ and the final one from Lemma \ref{thm:upper_bound_mi_divergence} and the independence of $\{ U_j \}$.

The probability  $p_{Z_i}(z_i) =  \mathop{\mathbb{E}}_{\{U_j\}} p_{Z_i  |\{U_{z_i}\}} $ is of the form
\begin{align}
    p_{Z_i} (z_i) =&   \mathop{\mathbb{E}}_{U_i} \mathrm{tr}\left(\left(\frac{(1 \pm \epsilon)}{L}\mathds{1} \mp \frac{2\epsilon}{L} U_i P U_i^\dagger\right) \rho \right) \nonumber \\
    =& \frac{1}{L}
\end{align}
For the other term, we have
\begin{align}
     \mathop{\mathbb{E}}_{U_i} p_{Z_i |  U_i}(z_i)^2 =& \mathop{\mathbb{E}}_{U_i} \mathrm{tr}\left(\left(\frac{(1 \pm \epsilon)}{L}\mathds{1} \mp  \frac{2\epsilon}{L} U_i P U_i^\dagger\right) \rho \right)^2  \nonumber\\
     &  {\leq }  \frac{1}{L^2} \left(1+ \frac{ {2} \epsilon^2}{(d+1)} \right).
\end{align}
Then, using Lemma~\ref{thm:upper_bound_mi_divergence}, we obtain
\begin{align}
    I(X: Y) &\leq \frac{N}{\ln(2)} \left(\sum_{z_i=1}^L \frac{1}{L} \left(1+ \frac{ {2} \epsilon^2}{(d+1)} \right) - 1\right) \nonumber\\
    &=\frac{ {2} N}{\ln(2)} \frac{\epsilon^2}{(d+1)}.
\end{align}
\end{proof}

\subsection{Lower bound on the sample complexity of $d_{\mathrm{av}}$}

\begin{theorem}
    Any procedure for quantum measurement tomography of a POVM on a $d$-dimensional Hilbert space that is $\epsilon/8$ accurate in average distance using nonadaptive, single-copy measurements on known input states requires
    \begin{align}
        N \in \Omega\left( \frac{d^2L}{\epsilon^2} \right)
    \end{align}
    uses of the unknown POVM. 
\end{theorem}

\begin{proof}
    Similar to the calculation of the lower bound for $d_{\mathrm{op}}, $ assume a random message is encoded in $N$ copies of a POVM $E_{\{U_j^X\}}$, where this POVM is uniformly sampled from the $\epsilon/4$-packing defined by Lemma \ref{thm:packing_dav}. Let $Y = (Y_1, \dots, Y_N)$ denote the outcomes from measuring the random POVM over $N$ quantum states. Since each POVM in the packing is separated by at least $\epsilon/4$, a tomography algorithm that takes this data and produces an estimate of the POVM within $\epsilon/8$ precision in $\sqrt{d} \, d_{\mathrm{av}}$ (with some constant probability) is sufficient to decode the message. Then, tomography must have a sample complexity at least as large as the discrimination problem. From the previous section, the mutual information between the measurements used for tomography and the message is upper bounded by 
\begin{align}
    I(X: Y) &\leq \frac{ {2} N}{\ln(2)} \frac{\epsilon^2}{(d+1)}.
\end{align}
Using Lemma \ref{thm:fano}, we also have 
\begin{align}
    I(X : Y) &\geq \Omega(d^2 L).
\end{align}
Combining these results, we obtain
\begin{align}
    \frac{ {2} N}{\ln(2)} \frac{\epsilon^2}{(d+1)} &\geq I(X: Y) \geq \Omega(d^2 L).
\end{align}
Thus, we conclude that
\begin{align}
    N &\geq \Omega \left( \frac{d^3 L}{\epsilon^2} \right)
\end{align}
to obtain precision $\epsilon$ in $\sqrt{d} d_{\mathrm{av}}$. Then, to obtain precision $\epsilon$ in  $d_{\mathrm{av}}$, we need at least
\begin{align}
    N &\geq \Omega \left( \frac{d^2 L}{\epsilon^2} \right). 
\end{align}
This matches the upper bound.
\end{proof}

\section{Experimental details}\label{app:experimental_setup}

The experiment is deployed on a subset of two connected qubits within a 5 qubit chip. It is a star topology flux-tunable superconducting transmon device manufactured by Quantware and hosted in the quantum computing lab of Technology Innovation Institute. The device is driven via microwave pulses generated via a Quantum Machines OPX system, orchestrated by the open source library Qibolab \cite{efthymiou2024qibolab}. The system is calibrated using the available routines on the open source Qibocal \cite{pasquale2024qibocal} package. 

The two qubits are calibrated up to usable specifications, as a reduced amount of noise is desirable for our presented algorithms to reconstruct the noisy POVMs. The two qubits are calibrated before launching the experiment up to the following parameters: the relaxation times for the first and second qubits are $T_1^{(1)}=33.(6)\,\mu s$ and $T_1^{(2)}=30.(6)\,\mu s$ respectively. Equally for the decoherence times $T_2^{(1)}=13.(0)\,\mu s$ and $T_2^{(2)}=9.(1)\,\mu s$. Assignment fidelity for a computational basis readout of $F_{ro}^{(1)}=96.(9)\,\%$ and $F_{ro}^{(2)}=98.(0)\,\%$. Lastly, the single qubit native pulse infidelity, as extracted from a randomized benchmarking experiment, of $g_{inf}^{(1)}=1.3(7)\cdot10^{-3}$ and $g_{inf}^{(2)}=6.(7)\cdot10^{-3}$ for native pulse length of $40\,ns$.

The experiments are designed directly in Qibolab language to leverage the lower level logic of the device. Everything is compiled down to the native pulses and interactions of the target quantum system. Code to reproduce the experiment is available upon reasonable request.

\section{Further results}\label{app:further_results}

\begin{figure}[t!]
    \includegraphics[width=0.99\linewidth]{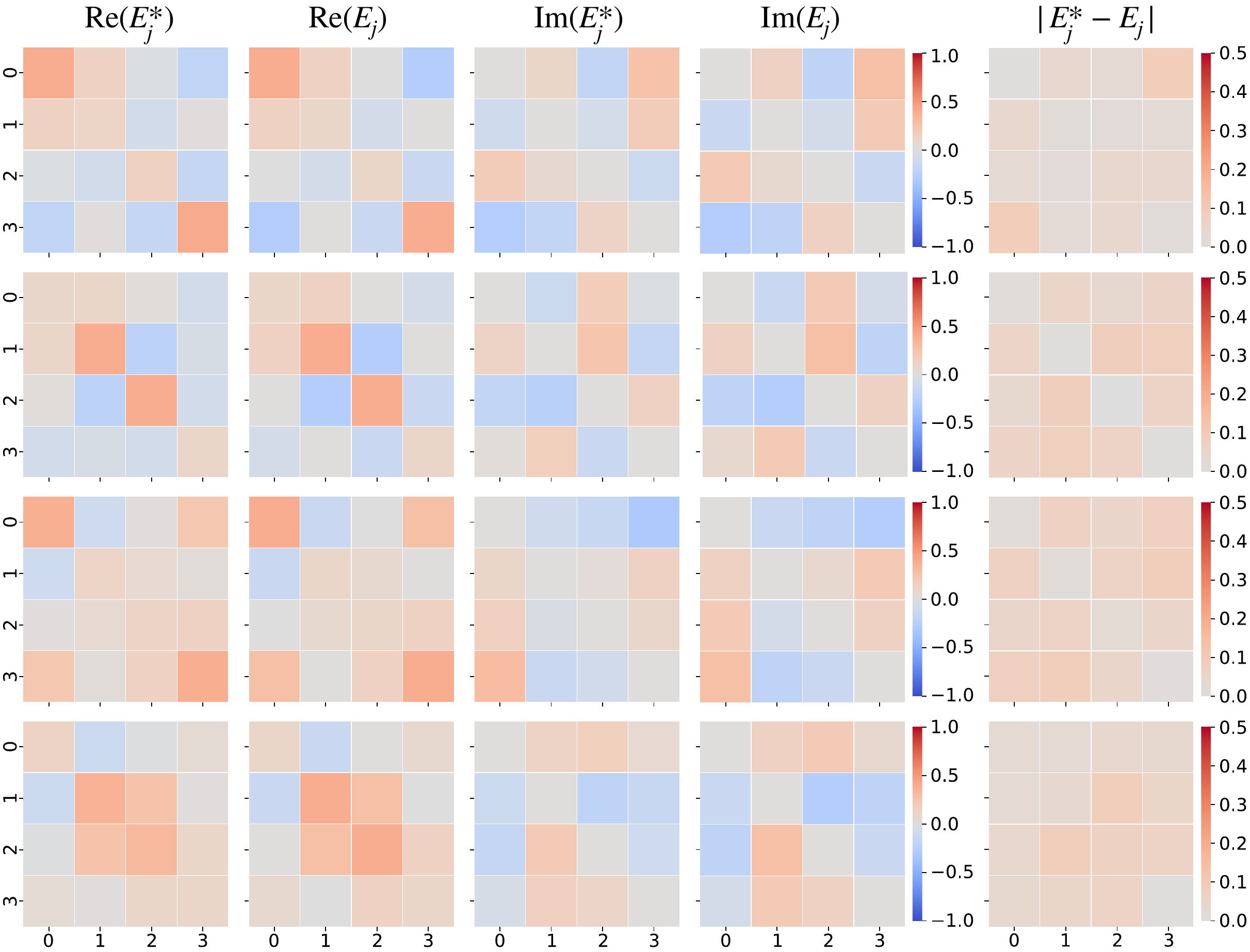}
    \centering
    \caption{Results for the reconstruction of the full system implementing a noisy two-qubit Symmetric Informationally Complete (SIC-)POVM, implemented on a two-qubit flux-tunable transmon device using a budget of $10^6$ random initial states. The reconstructed noisy POVM is compared to its expected noiseless counterpart, with the absolute difference displayed. Though noisy, the reconstruction follows the target POVM closely. The deviations from the expected values are due to the noisy device implementation.}
    \label{fig:app:sic-two}
\end{figure}

Here we expand on the reconstruction of noisy experimental POVM implementation on a two-qubit superconducting flux-tunable transmon system. First we give details on the implementation of the SIC-POVM presented in the main text, such as circuits and further reconstructions. After, we present results on other POVMs, such as the Identity, a Hadamard gates or Haar random gates on both qubits, as well as a Bell measurement.

\begin{figure}[b!]
    \centering
    \resizebox{0.7\linewidth}{!}{\quad\quad\input{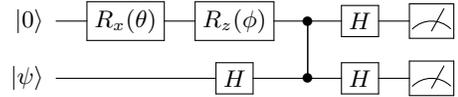}}
    \caption{Quantum circuit depicting the implementation of the SIC-POVM. When the single gate rotation parameters are $\theta=\arccos(1/\sqrt{3})$ and $\phi=3\pi/4$, this results in a Symmetric Informationally Complete measurement. The choice of a CZ gate as the entangling operation is due to the native interactions of the experimental device where this protocol is deployed.}
    \label{fig:app:siccircuit}
\end{figure}

We create the SIC-POVM with a quantum circuit of the form shown in Fig. \ref{fig:app:siccircuit}.
The data qubit, where the initial state will be prepared is the second qubit in Fig. \ref{fig:app:siccircuit}. The parameters for the single qubit rotations are taken as $\theta=\arccos(1/\sqrt{3})$ and $\phi=3\pi/4$ \cite{jiang2020sic_circuit} to generate a SIC-POVM \cite{you2024sic_basis}.
Furthermore, we also characterize this POVM as a two-qubit measurement, instead of a projection into a single qubit system. In Fig. \ref{fig:app:sic-two} we showcase the reconstruction. The results are similar to the ones in the main text, but in an extended space. Indeed, the single qubit POVM is located in the top left quadrant of the shown data.

In Fig \ref{fig:app:further-results} we collect all the data from the other POVM settings. First, on the top left are the results for the Identity POVM, that is, a direct measurement in the computational basis. While a very straightforward POVM, this reconstruction provides a great deal of information about the readout process of the device. For example, the non-symmetry in readout error when measuring different final states can be extracted from the data shown in this figure. On the top right, we plot the reconstruction results of a POVM measurement with a Hadamard gate before the readout, that is a measurement in the X basis --- where we consider the usual convention of the computational basis being the Z direction. The bottom-left panel showcases a similar scenario, but with a single qubit Haar random gate before the readout. Unlike the other instances, in this one we notice how the reconstructed POVM is randomly spread across the real and imaginary space, without a directly apparent structure. Finally, at the bottom right corner of Fig. \ref{fig:app:further-results} the results of the reconstruction of a noisy Bell POVM are shown. 

\begin{figure*}[t!]
    \includegraphics[width=0.99\linewidth]{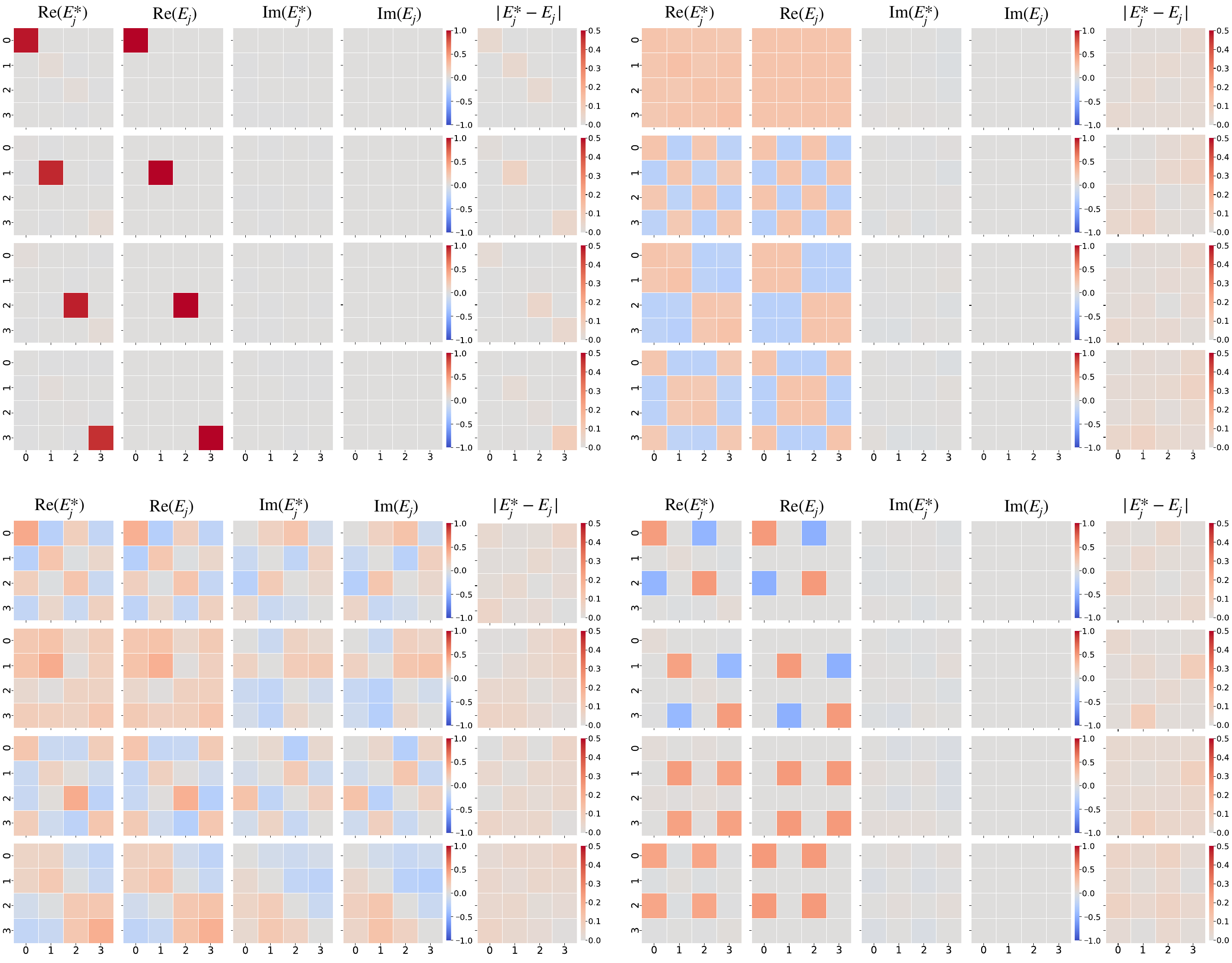}
    \centering
    \caption{Results for the reconstruction of various noisy two-qubit POVMs, implemented on a two-qubit flux-tunable transmon device using a budget of $10^6$ random initial states. The reconstructed noisy POVMs are compared to its expected noiseless counterpart, with the absolute difference displayed. Though noisy, the reconstruction follows the target POVMs closely. The deviations from the expected values are due to the noisy device implementation. More precisely, (top left) depicts the Identity POVM, where only a measurement in the computational basis is performed on both qubits. (top right) depicts a POVM where a Hadamard gate is applied to each qubit before measurement, that is, a change of basis from Z to X measurement basis. (bottom left) POVM where a Haar random gate is implemented before the measurement gate, we notice the how even with a more complex measurement the noisy implementation follows closely the noiseless simulation. Finally, on (bottom right) a Bell type POVM is reconstructed, where unlike the previous three POVMs depicted in this figure, an entangling gate is performed between the two qubits.}
    \label{fig:app:further-results}
\end{figure*}

\end{document}